\documentclass[aps,pra,twocolumn,groupedaddress,superscriptaddress,scrartcl,a4paper]{revtex4-1}
\usepackage{amssymb,amsmath,graphicx,epstopdf,hyperref,amsthm}
\newtheorem{lemma}{Lemma}
\newtheorem{proposition}{Proposition}

\begin{document}
\interfootnotelinepenalty=10000

\title{Hyperdense coding and superadditivity of classical capacities\\
in hypersphere theories}

\author{Serge Massar}
\author{Stefano Pironio}
\author{Dami\'an Pital\'ua-Garc\'ia}
\affiliation{Laboratoire d'Information Quantique, CP 224, Universit\'e libre de Bruxelles,
Av. F. D. Roosevelt 50, 1050 Bruxelles, Belgium}

\begin{abstract}
In quantum superdense coding, two parties previously sharing entanglement can communicate a two bit message by sending a single qubit. We study this feature in the broader framework of general probabilistic theories. We consider a particular class of theories in which the local state space of the communicating parties corresponds to Euclidean hyperballs of dimension $n$ (the case $n=3$ corresponds to the Bloch ball of quantum theory). We show that a single $n$-ball can encode at most one bit of information, independently of $n$. We introduce a bipartite extension of such theories for which there exist dense coding protocols such that $\log_2(n+1)$ bits are communicated if entanglement is previously shared by the communicating parties. For $n>3$, these protocols are more powerful than the quantum one, because more than two bits are communicated by transmission of a system that locally encodes at most one bit. We call this phenomenon hyperdense coding. Our hyperdense coding protocols imply superadditive classical capacities: two entangled systems can encode $\log_2(n+1)>2$ bits, even though each system individually encodes at most one bit. In our examples, hyperdense coding and superadditivity of classical capacities come at the expense of violating tomographic locality or dynamical continuous reversibility. 
\end{abstract}

\date{\today}
\maketitle

%\pacs{03.65.Ud}
%\keywords{}

\section{Introduction}

Classical information can be encoded in quantum systems and reliably recovered. One of the founding results in quantum information science \cite{NielsenandChuangbook} is the Holevo theorem \cite{K73}. It implies that, fundamentally, the classical capacity of $N$ qubits is $N$ bits: $N$ qubits can perfectly encode $N$ classical bits, but no more. It follows that classical capacities are additive in quantum theory. Though this seems quite natural -- it would be strange that by combining two systems that locally store one bit more than two bits could be encoded -- there exist, due to quantum entanglement, quantum channels whose capacities to communicate classical \cite{H09} or quantum \cite{SY08} information can be superadditive.

Entanglement is also responsible for other counter-intuitive aspects of the communication properties of quantum systems. In particular, it is at the basis of superdense coding \cite{sdc}, one of the fundamental protocols of quantum information theory. Though one qubit can locally encode at most one classical bit, in a superdense coding protocol a two bit message can, surprisingly, be communicated by the transmission of a single qubit with the aid of previously shared entanglement.

Though the use of entanglement therefore provides an advantage over purely local protocols, it has a limited communication power because no more than two bits can be communicated by transmission of a qubit independently of the amount of entanglement that the communicating parties share. The amount of quantum information that a transmitted quantum system can communicate is fundamentally limited by its Hilbert space dimension \cite{QIC}. This guarantees in particular that no more than $2N$ bits can be communicated in quantum superdense coding using a pair of entangled systems whose individual classical capacity is $N$ bits. A hypothetical protocol violating this quantum limit, which we denote a $\emph{hyperdense coding}$ protocol in the following, would imply a violation of the additivity property mentioned above.

Our motivation in this paper is to understand better the physical principles underlying the limited power of quantum superdense coding and the additivity of classical capacities. Could one conceive theories with hyperdense coding and superadditive classical capacities? How would such theories differ from quantum theory? What physical principles would they violate that are obeyed by quantum theory? Can the additivity of classical capacities and the inexistence of hyperdense coding be considered as conditions for physically sensible theories?

We address partially these questions in the framework of Generalised Probabilistic Theories (GPTs). These were introduced several decades ago~\cite{Mackey63,DL70,E70,Foulis81,Ludwig85}, and contain classical and quantum theory as special cases. This mathematical framework provides a basis for studying probabilistic theories that may have non-classical properties, and for formulating ``natural'' physical axioms for quantum theory \cite{H01,DB09,Chiribella10,Chiribella11,MM11,H11,H13,MMAP13}. The information properties of GPTs have been extensively studied \cite{B07,BBLW07,BBLW08,BBCLSSWW10,SW10,JGBB11,JL13,BW13,FMPT13,MP14,CHHH14,BKLS14,LB15}.

We consider a particular class of GPTs in which the state space of the communicating systems corresponds to an Euclidean ball of arbitrary dimension $n\in\mathbb{N}$. These \emph{hypersphere theories}, as we call them here, have been previously investigated \cite{U96,Z08,DB09,PDB10,MM11,PW12,MMPA14,MM13,MMAP13}. They have important physical motivations. The cases $n=1$ and $n=3$ correspond locally to the classical and the quantum bit, respectively. The local state space of a generalized bit can be deduced to be an $n-$ball from physically motivated axioms \cite{MMAP13}, which include a weak version of information causality \cite{ic}. 

We show that for arbitrary $n$, a single $n-$ball can  encode at most one bit of information. 
We then show that there exists a bipartite extension of hypersphere theories that includes entangled states
for which there exists a superdense coding protocol whose communication capacity is $\log_2 {(n+1)}$ bits. 
Thus, for $n>3$, this protocol is more powerful than the quantum one: it is hyperdense coding. Our hyperdense coding protocol imply that by using entangled states, two hypersphere systems can be used to encode $\log_2 {(n+1)}$ bits, thus achieving superadditive classical capacities when $n>3$.

We then turn to the interpretation of these results. We find a connection between superadditivity of classical capacities and hyperdense coding with two of the physical conditions imposed on GPTs in various derivations of finite dimensional quantum theory: continuous reversibility and tomographic locality \cite{H01,DB09,MM11,MMAP13}. 

This paper is organised as follows. We briefly introduce in Section~\ref{secgpt} the formalism of generalised probabilistic theories and in Section~\ref{secht} the specific class of hypersphere theories that we consider here. In Section~\ref{seccomm}, we investigate different communication scenarios in the context of these theories. We introduce a protocol for hyperdense coding, which implies the superadditivity of classical capacities, as well as for teleportation and entanglement swapping. We conclude with a discussion of the physical implications of our results. In particular, we present in the Appendices several additional results relating the (in)existence of hyperdense coding and superadditivity of classical capacities with physical properties such as tomographic locality and dynamical reversibility. 

\section{General Probabilistic Theories}\label{secgpt}
\subsection{States and effects}
In general probabilistic theories (GPTs), the space of unnormalised states is a proper cone $\mathcal{C}\subset\mathbb{R}^{n+1}$. The space of unnormalised effects is the dual cone of $\mathcal{C}$: $\mathcal{C}^*\equiv \lbrace e\in\mathbb{R}^{n+1} \vert  e\cdot\omega\geq 0, ~\forall \omega\in \mathcal{C}\rbrace$. The unit effect $u$ is an interior point of $\mathcal{C}^*$. A measurement is a set of effects that sum to the unit: $M=\{e_i\in \mathcal{C}^* \vert \sum_i e_i =u\}$. An unnormalised state $\omega\in\mathcal{C}$ has positive scalar product with the unit effect, $u\cdot \omega \geq 0$, while a normalised state $\omega \in \mathcal{C}$  has unit scalar product with the unit effect, $ u\cdot \omega  =1$. Given normalised state $\omega$ and measurement $M$, the probability of outcome $i$ is $p(i\vert M \omega)=  e_i\cdot \omega$, where `$\cdot$' denotes the Euclidean inner product. 

For any GPT, without loss of generality we can take the unit to have the form $u=\bigl(\begin{smallmatrix}
1\\ \bold{0}
\end{smallmatrix} \bigr)$, where $\bold{0}$ is the null vector in $\mathbb{R}^n$. The space of normalised states can then be written as
$\Omega\equiv \lbrace \omega_{r} \equiv \bigl(\begin{smallmatrix}
1\\ r
\end{smallmatrix} \bigr) \vert r\in \mathcal{R}\subset\mathbb{R}^n \rbrace $ with $\mathcal{R}$ a convex set. It is the convex hull of its extremal points, the pure states, which cannot be expressed as convex combinations of other states. The states in $\Omega$ that are not pure are called mixed. 

The unnormalised states can be expressed as $\omega'=\lambda\omega$ with $\lambda\geq 0$ and $\omega\in\Omega$ being a normalised state. The zero state is the null vector $\bold{0}\in\mathbb{R}^{n+1}$. The space of physical states is often defined as the set of unnormalised states $\omega'$ such that $u\cdot \omega'\leq 1$, which is the convex hull of the zero state and the set of normalised states $\Omega$. An unnormalised state corresponds to a preparation in which a particular outcome is obtained with a probability smaller than $1$. By considering all the outcomes obtained in a state preparation, we can consider that a state is always prepared with unit probability, hence, we can restrict the space of physical states to the set of normalised states $\Omega$. Here, we only consider normalised states $\Omega$.

The space of physically realisable effects will be noted $\mathcal{E}$. Any measurement corresponds to a set of  effects belonging to $\mathcal{E}$. We assume here that the set of all normalised effects is observable, that is, $\mathcal{E} = \lbrace e\in\mathcal{C}^* \vert   e\cdot \omega \leq 1, ~\forall \omega\in\Omega \rbrace$, which is convex. It can be argued that this does not need to hold in general, that $\mathcal{E}$ could as well be restricted to a proper subset of the normalised effects \cite{JL13}.

\subsection{Bipartite systems}

Let $\Omega_{A}\subset\mathbb{R}^{n_A+1}$, $\Omega_{B}\subset\mathbb{R}^{n_B+1}$ be the state spaces of the systems $A$ and $B$, and $\Omega_{AB}$ their joint state space. In general, $\Omega_{AB}$ can be defined arbitrarily. Unless otherwise stated, in the rest of this work we assume that a state $\phi\in\Omega_{AB}$ defines the outcome probabilities for each pair of effects $e_A \in \mathcal{E}_A$ and $e_B\in \mathcal{E}_B$, and that it satisfies two natural physical conditions: the no-signalling principle and tomographic locality \cite{H01,B07,BBLW07}. The \emph{no-signalling principle} states that the outcome probabilities for any measurement performed on the system $A$ are independent of what measurement is performed on the system $B$ and vice versa. \emph{Tomographic locality} states that all states $\phi\in\Omega_{AB}$ are characterized by the outcome probabilities of local measurements performed on $A$ and $B$.

It follows \cite{B07,BBLW07} from these conditions that $\Omega_{AB}\subset\mathbb{R}^{n_A+1}\otimes\mathbb{R}^{n_B+1}$ and that $\Omega_A\otimes_{\text{min}}\Omega_B\subseteq\Omega_{AB}\subseteq \Omega_A\otimes_{\text{max}}\Omega_B$, where $\Omega_A\otimes_{\text{min}}\Omega_B \equiv\text{convex hull}\{\omega_A\otimes\omega_B\vert \omega_A\in\Omega_A,\omega_B\in\Omega_B\}$ is the \emph{minimal tensor product} and  $\Omega_A\otimes_{\text{max}}\Omega_B \equiv\{\phi\vert (u_{A}\otimes u_B)\cdot\phi=1, (e_A\otimes e_B)\cdot\phi\geq0~ \forall e_A\in\mathcal{E}_A,e_B\in\mathcal{E}_B\}$ is the \emph{maximal tensor product}. Similar definitions can be given for the space of effects: $\mathcal{E}_A\otimes_{\text{min}}\mathcal{E}_B \equiv\text{convex hull}\{e_A\otimes e_B\vert e_A\in\mathcal{E}_A,e_B\in\mathcal{E}_B\}$ and  $\mathcal{E}_A\otimes_{\text{max}}\mathcal{E}_B \equiv\{E\vert  0\leq E\cdot(\omega_A\otimes \omega_B)\leq 1 ~ \forall \omega_A\in\Omega_A,\omega_B\in\Omega_B\}$. The minimal tensor product $\Omega_A\otimes_{\text{min}}\Omega_B$ includes all the separable states but does not contain any entangled states. The maximal tensor product $\Omega_A\otimes_{\text{max}}\Omega_B$ contains all states that are consistent with the no-signalling principle and that give valid probabilities to all local measurements. The states in $\Omega_A\otimes_{\text{max}}\Omega_B$ that are not in $\Omega_A\otimes_{\text{min}}\Omega_B$ are entangled.  If $A$ or $B$ is classical then $\Omega_{AB}=\Omega_A\otimes_{\text{min}}\Omega_B=\Omega_A\otimes_{\text{max}}\Omega_B$, while for quantum theory $\Omega_A\otimes_{\text{min}}\Omega_B\subset \Omega_{AB}\subset \Omega_A\otimes_{\text{max}}\Omega_B$ \cite{BBLW07}.

We define the unit effect of $\Omega_{AB}$ as $u_{AB}\equiv u_A\otimes u_B$. It follows that all states $\phi\in\Omega_{AB}$ can be expressed as matrices
\begin{equation}
\label{eq:011}
\phi\equiv 
 \begin{pmatrix}
1 & b^{\text{t}} \\
 a  & C
 \end{pmatrix},
\end{equation}
where 
$a\in\mathbb{R}^{n_A}$, $b\in\mathbb{R}^{n_B}$ and $C\in\mathbb{R}^{n_A}\otimes\mathbb{R}^{n_B}$. Similarly, all effects $E\in\mathcal{E}_{AB}$ can be expressed as matrices
\begin{equation}
\label{eq:012}
 E=\begin{pmatrix}
 \gamma& \beta^{\text{t}} \\
 \alpha  & \Gamma
 \end{pmatrix},
\end{equation}
where 
$\gamma\in\mathbb{R}$, $\alpha\in\mathbb{R}^{n_A}$, $\beta\in\mathbb{R}^{n_B}$ and $\Gamma\in\mathbb{R}^{n_A}\otimes\mathbb{R}^{n_B}$. The probability of obtaining outcome given by effect $E$ if the state is $\phi$ is
$E\cdot\phi=\text{Tr}(E^{\text{t}}\phi)$, where `t' denotes transposition. The unit effect $u_{AB}$ is a matrix (\ref{eq:012}) with $\gamma=1$ and all other entries zero. The reduced states $\omega_A\equiv\phi u_B=\bigl(\begin{smallmatrix}
1\\ a
\end{smallmatrix} \bigr)$ and $\omega_B\equiv\phi^\text{t}u_A=\bigl(\begin{smallmatrix}
1\\ b
\end{smallmatrix} \bigr)$ must satisfy $\omega_A\in\Omega_A$ and $\omega_B\in\Omega_B$. Thus, it is required that $a\in\mathcal{R}_A$ and $b\in\mathcal{R}_B$.

\subsection{Dynamics}

A GPT specifies the state spaces $\Omega_A$, $\Omega_B$ and $\Omega_{AB}$ for all physical systems $A$ and $B$. It also specifies the set of allowed transformations on the state spaces. We adopt the following consistency condition as the minimal physical condition that the allowed transformations must satisfy \cite{B07}. The set of allowed transformations $\mathcal{T}_A$ on system $A$ is \emph{consistent} if for every $T_A\in\mathcal{T}_A$, we have $T_A:\Omega_A\to\Omega_A$ and $T_A\otimes I_B:\Omega_{AB}\to\Omega_{AB}$, where $\Omega_{AB}$ is the joint state space of $A$ and any system $B$, and $I_{B}$ is the identity map on $B$. 

In general, the allowed transformations can be represented as linear maps \cite{B07}. Thus, we express any transformation $T_A\in\mathcal{T}_A$ as a matrix in $\mathbb{R}^{n_A+1}\otimes\mathbb{R}^{n_A+1}$.

We notice that according to the conditions given above, the allowed transformations must be normalisation-preserving, that is, normalised states are transformed into normalised states. It is often considered that transformations $T$ can output an unnormalised state $\omega'=T\omega$, with $\omega\in\Omega$ such that $0\leq u\cdot \omega'<1$. These are normalisation-decreasing transformations and correspond to obtaining an outcome with probability less than $1$. However, by considering all the possible outcomes, the transformations reduce to normalisation-preserving ones. Here, we only consider normalisation-preserving transformations, as given by the consistency condition.

\section{Hypersphere theories}\label{secht}

We introduce in this section a family of GPTs that we call hypersphere theories. These
theories have been studied before \cite{U96,Z08,DB09,PDB10,MM11,PW12,MMAP13,MMPA14,MM13}.
The state space of single systems in \emph{hypersphere theories} (HSTs) is defined as the unit ball of dimension $n$, $\Omega\equiv \bigl\{\omega_r\equiv \bigl(\begin{smallmatrix}
1\\ r
\end{smallmatrix} \bigr) \big\vert r\in \mathcal{R}\bigr\}$ with $\mathcal{R}\equiv\bigl\{ r \in \mathbb{R}^n \big\vert \lVert  r \rVert \leq 1\bigr\}$, for which the set of pure states $\lVert r \rVert =1$ defines a hypersphere. The extremal effects are $e_m=\frac{1}{2} \bigl(\begin{smallmatrix}
1\\ m
\end{smallmatrix} \bigr)$ with $m\in\mathbb{R}^n$ and $\lVert m\rVert =1$, defining also a hypersphere. Each unit vector $\lVert m \rVert =1$ defines a canonical measurement: $M_m=\{e_ m,e_{- m}\}$.

\subsection{Bipartite systems}
We now introduce a particular extension of HSTs to two systems that, as we show in the next section, has hyperdense coding, superadditive classical capacities, teleportation and entanglement swapping. We leave it as an open question to investigate other possible extensions of HSTs, in bipartite and multipartite settings, and their corresponding communication properties.

Before doing so, we introduce a useful algebraic structure (based on Hadamard transformations). Let $\mu\equiv(\mu_0,\mu_1,\ldots,\mu_{N-1})\in\lbrace 0,1\rbrace^N$ and $\nu\equiv(\nu_0,\nu_1,\ldots,\nu_{N-1})\in\lbrace 0,1\rbrace^N$ be $N$ bit strings. We define $2^N$ vectors $d_{\mu}\in\mathbb{R}^{2^N}$, labelled by $\mu$, whose components, labelled by $\nu$, are
\begin{equation}
\label{eq:43.0}
(d_{\mu})_{\nu}=(-1)^{\mu\cdot\nu},
\end{equation}
where $\mu\cdot \nu=\oplus_{l=0}^{N-1}\mu_l\nu_l$, with $\oplus$ denoting sum modulo $2$ (we note that $(d_{\mu})_{\nu}$ is the $\mu,\nu$ entry of a real Hadamard matrix). We define this set of vectors as 
$\mathcal{D}_N\equiv \bigl\{d_{\mu}\}_{\mu\in\{0,1\}^N}$.

It is easy to see that $(d_{\bold{0}})_{\nu}=(d_\mu)_{\bold{0}}=1$ for all $\mu,\nu\in\{0,1\}^N$, where $\bold{0}\equiv(0,0,\ldots,0)$ is the string with $N$ zero entries. Additionally, for $\mu\neq\bold{0}$, we have that $d_{\mu}$ has $2^{N-1}$ entries equal to $-1$, and the other $2^{N-1}$ entries equal to $+1$. To see this, suppose that $\mu$ has exactly $k$ entries equal to $1$, with $1\leq k \leq N$, in positions $l_1,\l_2,\ldots,l_k$. Thus, we have $\mu_{l_1}=\mu_{l_2}=\cdots=\mu_{l_k}=1$, with all other entries equal to zero. It follows from (\ref{eq:43.0}) that $(d_\mu)_\nu=-1$ iff $\oplus_{i=1}^{k}\nu_{l_i}=1$, where the sum is modulo two, which holds for exactly half of the possible values of $\nu$, that is, for $2^{N-1}$ cases. Thus, $d_{\mu}$ has an even number of $-1$ entries for $N\geq 2$ and $\mu\neq \bold{0}$.

Furthermore, it follows easily from (\ref{eq:43.0}) that $\mathcal{D}_N$ is closed under the element-wise product:
\begin{equation}
\label{eq:43.01}
d_\mu\circ d_{\mu'}=d_{\mu\oplus\mu'}\in\mathcal{D}_N,
\end{equation}
for all $\mu,\mu'\in\{0,1\}^N$. That is, $(d_\mu\circ d_{\mu'})_{\nu}\equiv (d_\mu)_{\nu} (d_{\mu'})_{\nu} =(d_{\mu\oplus\mu'})_\nu$ for all $\mu,\mu',\nu\in\{0,1\}^N$, where $\mu\oplus\mu'\equiv(\mu_0\oplus \mu_0',\mu_1\oplus \mu_1',\ldots,\mu_{N-1}\oplus \mu_{N-1}')$.
Thus $\mathcal{D}_N$ is a group under the element-wise product, with $d_{\bold{0}}$ being the identity element. Moreover, the elements of $\mathcal{D}_N$ satisfy:
\begin{eqnarray}
\label{eq:43.2}
\sum_{{\mu}\in\{0,1\}^N} (d_{\mu})_{\nu}&=&2^N \delta_{\nu,\bold{0}},\\
\label{eq:43.3}
d_\mu\cdot d_{\mu'} \equiv\sum_\nu (d_{\mu})_{\nu}(d_{\mu'})_{\nu} &=&2^N \delta_{\mu,\mu'}.
\end{eqnarray}

Let us now define the set $\Omega_{AB}$ of bipartite states for systems $A$ and $B$ that are locally described by HST systems of dimension $n=2^N-1$. We take 
$$\Omega_{AB}
\equiv\text{convex hull }\{\Omega_A\otimes_{\text{min}}\Omega_B,\mathcal{S}_N\}$$ to be the convex hull of the product states and the discrete set of entangled states $\mathcal{S}_N\equiv\{\phi_\mu\}_{\mu\in\{0,1\}^N}$ defined as the set of diagonal matrices
\begin{equation}
\label{eq:43.310}
\phi_\mu\equiv \text{diag}(d_\mu),
\end{equation}
whose entries are given by $(\phi_\mu)_{\nu,\nu'}\equiv\delta_{\nu,\nu'}(d_{\mu})_{\nu}$, where $\nu,\nu'\in\{0,1\}^N$.
The corresponding reduced states for systems $A$ and $B$ are the completely ``mixed" states $\omega_A\equiv\phi_{\mu} u_B=\bigl(\begin{smallmatrix}
1\\ \bold{0}
\end{smallmatrix} \bigr)\in\Omega_A$ and $\omega_B\equiv\phi_\mu^\text{t}u_A=\bigl(\begin{smallmatrix}
1\\ \bold{0}
\end{smallmatrix} \bigr)\in\Omega_B$, with $\bold{0}$ being the null vector in $\mathbb{R}^{2^N-1}$.

We introduce a discrete set of entangled effects $\mathcal{F}_N\equiv\{E_\mu\}_{\mu\in\{0,1\}^N}$, through
\begin{equation}
\label{eq:43.FN}
E_\mu\equiv 2^{-N} \phi_\mu.
\end{equation}
This set defines a measurement 
\begin{equation}
\label{eq:43.M}
M=\{E_\mu\}_{ \mu\in\{0,1\}^N},
\end{equation}
because 
\begin{equation}
\label{eq:lasteqever}
\sum_{\mu\in\{0,1\}^N} E_\mu=\text{diag}(1,\bold{0})=u_{AB},
\end{equation}
as follows from (\ref{eq:43.2}), (\ref{eq:43.310}) and (\ref{eq:43.FN}).
The probability to get outcome $\mu$ when measuring the entangled state $\phi_{\mu'}$ is
\begin{equation}
\label{eq:43.33}
E_\mu\cdot\phi_{\mu'}=2^{-N}d_\mu\cdot d_{\mu'}=\delta_{\mu,\mu'},
\end{equation}
as can be deduced from (\ref{eq:43.3}) -- (\ref{eq:43.FN}).

Finally, we define the allowed local transformations as the convex hull of the discrete set $\mathcal{T}_N\equiv\{T_\mu\}_{\mu\in\{0,1\}^N}$ where
\begin{equation}
\label{eq:43.51}
T_\mu\equiv \phi_\mu.
\end{equation}
Note that $T_{\bold{0}}=I$ is the identity in $\mathbb{R}^{2^N}$.
Furthermore, from (\ref{eq:43.01}), (\ref{eq:43.310}) and (\ref{eq:43.51}), we have
\begin{equation}
\label{eq:43.52}
T_{\mu'}T_{\mu}=T_{\mu\oplus\mu'}\in\mathcal{T}_N,
\end{equation}
for all $\mu,\mu'\in\{0,1\}^N$. Thus, we have that $\mathcal{T}_N$ is a group under matrix multiplication. Since the vectors $d_\mu$ have only $\pm 1$ entries with an even number of $-1$ entries for $N\geq 2$, it follows from (\ref{eq:43.310}) and (\ref{eq:43.51}) that $\mathcal{T}_N$ is a subgroup of $\text{SO}(2^N)$ for $N\geq 2$.

Since $(T_\mu)_{\bold{0},\bold{0}}=(d_\mu)_{\bold{0}}=1$ for all $\mu\in\{0,1\}^N$, we can express the matrices $T_{\mu}$ as
\begin{equation}
\label{eq:43.53}
T_{\mu} = \begin{pmatrix}
1 & \bold{0} \\
 \bold{0}  & \hat{T}_\mu
 \end{pmatrix}.
\end{equation}
 It follows that $\hat{\mathcal{T}}_N\equiv \{\hat{T}_{\mu}\}_{\mu\in\{0,1\}^N}$ is a subgroup of $\text{SO}(2^N-1)$ for $N\geq 2$, and $\hat{\mathcal{T}}_1=\text{O}(1)$.

\subsection{Consistency of the above definitions}
We show that the above definitions make sense in the GPT framework considered here, i.e., that the set $\Omega_{AB}$ of bipartite states is contained in the maximal tensor product $\Omega_A\otimes_{\text{max}}\Omega_B$, that the set of effects $\mathcal{F}_N$ is contained in the space of normalized effects $\mathcal{E}_{AB}$ of $\Omega_{AB}$, and that the local transformations $\mathcal{T}_N$ are consistent in the sense that $T_\mu: \Omega_A\rightarrow \Omega_A$ and $(T_\mu\otimes I_B): \Omega_{AB}\rightarrow \Omega_{AB}$ for all $T_\mu\in\mathcal{T}_N$. Additionally, we show that tomographic locality is satisfied.

Let us first show that $\mathcal{T}_N$ satisfies the consistency condition. Consider $\omega_a\in\Omega_A$. From (\ref{eq:43.53}), we have $T_\mu\omega_a=\bigl(\begin{smallmatrix}
1\\ a_\mu
\end{smallmatrix} \bigr)$, where $a_\mu\equiv \hat{T}_\mu a$. Since $\hat{T}_\mu\in\text{O}(2^N-1)$ and $\omega_a\in\Omega_A$, we have $\lVert a_\mu\rVert =\lVert a\rVert \leq 1$. Thus, $T_\mu\omega_a\in\Omega_A$. Consider now $\phi\in\Omega_{AB}$. If $\phi\in\Omega_A\otimes_{\text{min}}\Omega_B$, it follows straightforwardly that $T_\mu\phi\in\Omega_{AB}$. If $\phi$ is one of the $\phi_{\mu'}\in\mathcal{S}_N$, we have from (\ref{eq:43.51}) and (\ref{eq:43.52}) that $T_\mu\phi_{\mu'}=\phi_{\mu\oplus\mu'} \in\mathcal{S}_N\in\Omega_{AB}$ for any $\phi_{\mu'}\in \mathcal{S}_N$. 

Let us now show that $\mathcal{S}_N\subset\Omega_A\otimes_{\text{max}}\Omega_B$. This implies from the definition of $\Omega_{AB}$ that $\Omega_A\otimes_{\text{min}}\Omega_B\subset\Omega_{AB}\subseteq \Omega_A\otimes_{\text{max}}\Omega_B$. Let $\phi_\mu\in\mathcal{S}_N$ and consider arbitrary extremal effects  $e_{\alpha}\equiv \frac{1}{2}\bigl(\begin{smallmatrix}
1\\ \alpha
\end{smallmatrix} \bigr)\in\mathcal{E}_A$ and $e_{\beta}\equiv \frac{1}{2}\bigl(\begin{smallmatrix}
1\\ \beta
\end{smallmatrix} \bigr)\in\mathcal{E}_B$, where $\alpha,\beta\in\mathbb{R}^{2^N-1}$ and $\lVert \alpha\rVert=\lVert \beta\rVert=1$. We have that
\begin{equation}
\label{eq:43.31}
(e_{\alpha}\otimes e_{\beta}) \cdot\phi_\mu= \frac{1}{4}\bigl(\begin{smallmatrix}
1\\ \alpha
\end{smallmatrix} \bigr) \cdot \bigl(\begin{smallmatrix}
1\\ \beta_\mu
\end{smallmatrix} \bigr)=\frac{1}{4}(1+\alpha\cdot \beta_\mu),
\end{equation}
where $\beta_\mu\equiv \hat{T}_\mu \beta$. Since $\hat{T}_\mu\in \text{O}(2^N-1)$, we have $\lVert \beta_\mu\rVert = \lVert \beta \rVert = 1$. Thus, $0\leq (e_{\alpha}\otimes e_{\beta}) \cdot\phi_\mu \leq \frac{1}{2}< 1$. Furthermore, from (\ref{eq:43.0}) and (\ref{eq:43.310}),we have that $(u_A\otimes u_B)\cdot \phi_\mu=(d_\mu)_{\bold{0}}=1$, which proves the claim.

In addition, tomographic locality is satisfied. It is straightforward to see that the entries $(\phi)_{\nu,\nu'}$ of the states $\phi\in\Omega_{AB}$ are determined by the local measurements $M_A=\bigl\lbrace \frac{1}{2}\bigl(\begin{smallmatrix}
1\\ v_\nu
\end{smallmatrix} \bigr), \frac{1}{2}\bigl(\begin{smallmatrix}
1\\ -v_\nu
\end{smallmatrix} \bigr)\bigr\rbrace$ and $M_B=\bigl\lbrace \frac{1}{2}\bigl(\begin{smallmatrix}
1\\ v_{\nu'}
\end{smallmatrix} \bigr), \frac{1}{2}\bigl(\begin{smallmatrix}
1\\ -v_{\nu'}
\end{smallmatrix} \bigr)$ on systems $A$ and $B$, respectively, where $v_\nu\in\mathbb{R}^{2^N-1}$ is a vector whose $\nu$th entry equals unity and whose other entries are zero.

Finally, we verify that $\mathcal{F}_N\subset\mathcal{E}_{AB}$. We need to show that, for any $E_\mu\in \mathcal{F}_N$, we have: i) $0\leq E_\mu\cdot(\omega_a\otimes\omega_b)\leq 1$ for arbitrary pure states $\omega_a\equiv\bigl(\begin{smallmatrix}
1\\ a
\end{smallmatrix} \bigr)\in\Omega_A$ and $\omega_b\equiv\bigl(\begin{smallmatrix}
1\\ b
\end{smallmatrix} \bigr)\in\Omega_B$, where $a,b\in\mathbb{R}^{2^N-1}$ and $\lVert a\rVert=\lVert b\rVert=1$; and ii) $0\leq E_\mu\cdot\phi_{\mu'}\leq 1$ for arbitrary entangled states $\phi_{\mu'}\in \mathcal{S}_N$.

We show i). From (\ref{eq:43.FN}), (\ref{eq:43.51}) and (\ref{eq:43.53}), we have
\begin{equation}
E_\mu\cdot(\omega_a\otimes\omega_b)=2^{-N}\phi_{\mu}\cdot(\omega_a\otimes\omega_b)=2^{-N}(1+a\cdot b_{\mu}),
\end{equation}
where $b_\mu\equiv \hat{T}_{\mu} b$. Since $\hat{T}_{\mu}\in \text{O}(2^N-1)$, we have $\lVert b_\mu\rVert = \lVert b \rVert = 1$. Thus, $0\leq E_\mu \cdot(\omega_a\otimes \omega_b)  \leq 2^{-(N-1)}\leq 1$, as claimed.

We show ii). This is implied by (\ref{eq:43.33}).

\section{Communication in GPTs and hypersphere theories}\label{seccomm}

\subsection{Classical capacities}
Consider a situation where Alice sends Bob a classical message $x$ with probability $p(x)$ by encoding it in a GPT state $\omega_x$ and where Bob decodes the message using a decoding measurement $M=\{e_y \in \mathcal{E} \vert \sum_y e_y=u\}$. Given that the message $x$ was sent, Bob obtains the outcome $y$ with probability $p(y|x)=e_y\cdot\omega_x$. The mutual information $I(X:Y)=\sum_{xy} p(xy)\log_2\Bigl( p(xy)/\bigl(p(x)p(y)\bigr)\Bigr)$ quantifies the amount of classical information that is transmitted through such a protocol.

\theoremstyle{definition}
\newtheorem*{CCGPT}{Classical capacity of a GPT}
\begin{CCGPT}
\label{CCGPT}
The classical capacity $\chi_{\text{C}}(\Omega)$ of a GPT with state space $\Omega$ 
is the maximum of $I(X:Y)$ over all probability distributions $p(x)$, encoding states $\omega_x\in\Omega$, and decoding measurements $M$. 
\end{CCGPT}

The classical capacity of $d$-dimensional classical theory (a classical dit) is $\log_2 d$. Similarly, the classical capacity of a qudit is $\log_ 2 d$, as follows from the Holevo bound \cite{K73}. The following proposition allows one to put upper bounds on the classical capacities of GPTs. It will prove useful below.

\begin{proposition}
\label{theorem1A}
Consider a GPT with state space $\Omega\equiv \bigl\lbrace \omega_r \equiv \bigl(\begin{smallmatrix}
1\\ r
\end{smallmatrix} \bigr)\big\vert r\in \mathcal{R} \subset \mathbb{R}^n\bigr\rbrace $ and unit effect $u\equiv \bigl(\begin{smallmatrix}
1\\ \bold{0}
\end{smallmatrix} \bigr)$. Assume that any effect $e\in\mathcal{E}$ can be expresed as $e=\gamma\bigl(\begin{smallmatrix}
1\\ m
\end{smallmatrix} \bigr)$, with $\gamma\in[0,1]$ and $m\in\mathcal{M}\subset \mathbb{R}^n$. 
Then the classical capacity of the GPT $\Omega$ is bounded by
\begin{equation}
\label{eq:5}
\chi_{C}(\Omega)\leq \max_{r \in \mathcal{R}, m\in \mathcal{M}} \bigl\lbrace \log_2 (1+ m\cdot r )\bigr\rbrace\leq  \log_2 (1+  MR),\nonumber
\end{equation}
where $M\equiv \max_{m\in\mathcal{M}}\lbrace {\lVert m \rVert}\rbrace$ and $R\equiv \max_{r\in\mathcal{R}}\lbrace{\lVert r \rVert}\rbrace$ .
\end{proposition}
\begin{proof}
Let $X$ be the random variable corresponding to messages $x$, chosen with probability $p(x)=p_x$, encoded in states $\omega_x \equiv \bigl(\begin{smallmatrix}
1\\ r_x
\end{smallmatrix} \bigr)$,  where $r_x\in\mathcal{R}$. Let $Y$ be the random variable of the measurement outcomes $y$ corresponding to effects $ e_y\equiv \gamma_y\bigl(\begin{smallmatrix}
1\\ m_y
\end{smallmatrix} \bigr)$, where $\gamma_y\in[0,1]$ and $m_y\in\mathcal{M}$. 
The condition $\sum_y e_y=u\equiv\bigl(\begin{smallmatrix}
1\\ \bold{0}
\end{smallmatrix} \bigr)$ implies
\begin{equation}
\label{eq:6.3}
\sum_y \gamma_y =1,\qquad
\sum_{y}\gamma_y m_y = \bold{0}. 
\end{equation}
The outcome probabilities are
\begin{equation}
\label{eq:6.1}
p(y\vert x) = e_y\cdot \omega_x=\gamma_y(1+ m_y \cdot r_x).
\end{equation}
The condition $p(y\vert x)\geq 0$ implies that
\begin{equation}
\label{eq:6.2}
m_y \cdot r_x \geq -1,
\end{equation}
for $\gamma_y>0$. Notice that the case $\gamma_y=0$ corresponds to probabilities $p(y\vert x)=0$ which have null contributions to the Shannon entropies $H(Y\vert X)$ and $H(Y)$, and hence to the mutual information $I(X:Y)$. Thus, without loss of generality we consider that $\gamma_y>0$, for which (\ref{eq:6.2}) holds.

 From the definition of the mutual information $I(X:Y)=H(Y)-H(Y\vert X)$ and (\ref{eq:6.1}), it is straightforward to obtain the expression
\begin{equation}
\label{eq:6.5}
I(X:Y)=L_1 - L_2,
\end{equation}
where 
\begin{eqnarray}
\label{eq:6.6}
L_1\!&\equiv&\! \sum_x\! p_x\!\sum_y\! \gamma_y(1+m_y\!\cdot\! r_x)\log_2(1+m_y\!\cdot\! r_x),\\
\label{eq:6.7}
L_2\!&\equiv&\! \sum_y\! \gamma_y(1+m_y\!\cdot\! \bar{r})\log_2(1+m_y\!\cdot\! \bar{r}),
\end{eqnarray}
and $ \bar{r}\equiv{\sum_xp_xr_x}$. 

Let $D\equiv \max_{r_x \in \mathcal{R}, m_y\in \mathcal{M}} \bigl\lbrace \log_2 (1+ m_y\cdot r_x )\bigr\rbrace$. We have from (\ref{eq:6.6}),  the definition of $D$, and (\ref{eq:6.3}), that
\begin{eqnarray}
\label{eq:6.10}
L_1&\leq& \sum_xp_x\sum_y\gamma_y(1+m_y\cdot r_x) D,\nonumber\\
\quad &=&\Bigl(\sum_y\gamma_y+ \sum_y\gamma_ym_y\cdot \bar{r}\Bigr)D,\nonumber\\
\quad &=&D,
\end{eqnarray}

Denote $z_y=1+m_y\cdot\bar{r}$. We have $z_y\geq 0$ and $\sum_y \gamma_y z_y =1$ with $\gamma_y\geq 0$ a probability distribution since $\sum_y \gamma_y =1$. Hence we have
 \begin{eqnarray}
\label{eq:6.11}
L_2&=& \sum_y \gamma_y z_y \log_2 z_y,\nonumber\\
\quad &\geq & \left ( \sum_y \gamma_y z_y \right) \log_2 \left ( \sum_y \gamma_y z_y \right) ,\nonumber\\
&=& 0
\end{eqnarray}
where we have used convexity of $x\log_2 x$ for $x\geq 0$.
\end{proof}

\begin{proposition}
\label{theorem1}
The classical capacity of  hypersphere theories is equal to $1$ bit.
\end{proposition}

\begin{proof}
It is immediate to check that the channel in which Alice prepares pure states $\omega_ r$ and $\omega_{- r}$, with $\lVert r \rVert =1$, each with probability $\frac{1}{2}$,
and Bob carries out measurement $M=\bigl\{\frac{1}{2}\bigl(\begin{smallmatrix}
1\\ r
\end{smallmatrix} \bigr) , \frac{1}{2}\bigl(\begin{smallmatrix}
1\\-r
\end{smallmatrix} \bigr)\bigr\}$, has capacity $1$ bit. Hence the capacity of hypersphere theories is at least $1$ bit.

The converse follows from Proposition \ref{theorem1A}, as in the case of hypersphere theories, we have $M=R=1$.
\end{proof}

\theoremstyle{definition}
\newtheorem*{SACC}{Superadditivity of classical capacities}
\begin{SACC}
\label{SACC}
The classical capacity $\chi_{\text{C}}(\Omega_{AB})$ of a GPT with state space $\Omega_{AB}$ is superadditive if $\chi_{\text{C}}(\Omega_{AB})>\chi_{\text{C}}(\Omega_A)+\chi_{\text{C}}(\Omega_B)$.
\end{SACC}
The classical capacities of classical and quantum theory are additive: $\chi_{\text{C}}(\Omega_{AB})=\chi_{\text{C}}(\Omega_A)+\chi_{\text{C}}(\Omega_B)$. We show below that the hypersphere theories defined in section \ref{secht} have superadditive classical capacities.

\subsection{Dense coding}

The classical capacity of a GPT is the maximum classical information that can be transmitted without the assistance of previously shared resources. The classical capacity can sometimes be enhanced by the use of previously shared entanglement, as in the following general dense coding protocol.

In a dense coding protocol, Alice and Bob initially share a bipartite state $\phi\in\Omega_{AB}$. Alice chooses a message $x$ with probability $p_x$ from a finite set and applies the local transformation $T_x\in\mathcal{T}_A$ on her system $A$. After Alice's operation, the state is transformed into $\phi_x\equiv T_x \phi$. Alice sends Bob her system. After receiving system $A$, Bob applies a measurement on the composite system $AB$, defined by the set of effects $\lbrace E_y\in\mathcal{E}_{AB}\vert \sum_yE_y=u_{AB}\rbrace$, obtaining the outcome $y$ with probability $p(y\vert x) = E_y\cdot\phi_x$.

Note that following our notation given by (\ref{eq:011}) and (\ref{eq:012}), $E_y$ and $\phi_x$ are real matrices, $T_x\phi$ denotes matrix multiplication, and $E_y\cdot\phi_x=\text{Tr}(E_y^{\text{t}}\phi_x)$, where `t' denotes transposition.

\theoremstyle{definition}
\newtheorem*{DCC}{Dense coding capacity of a state}
\begin{DCC}
\label{DCC}
The dense coding capacity $\chi_{\text{DC}}(\phi)$ of a GPT state $\phi\in\Omega_{AB}$ is $\chi_{\text{DC}}(\phi)=\max I(X:Y)$ where the max is taken over all \hyperref[DC]{dense coding} protocols in which the initially shared state is $\phi$ and where $I(X:Y)$ is the mutual information between $x$ and $y$.
\end{DCC}

\theoremstyle{definition}
\newtheorem*{TDCC}{Dense coding capacity of a GPT}
\begin{TDCC}
\label{TDCC}
The dense coding capacity $\chi_{\text{DC}}(\Omega_{AB})$ of a GPT with state space $\Omega_{AB}$
is the maximum of $\chi_{\text{DC}}(\phi)$ over all states $\phi\in\Omega_{AB}$.
\end{TDCC}

Note that we trivially have the following inequality, since a dense coding protocol can be viewed as a communication protocol in which the states $T_x\phi$ are prepared:
\begin{proposition}
\label{theochichi}
The \hyperref[CCGPT]{classical} and \hyperref[TDCC]{dense coding} capacities of a GPT with bipartite state space $\Omega_{AB}$ satisfy $\chi_{\text{DC}}(\Omega_{AB}) \leq \chi_{\text{C}}(\Omega_{AB})$.
\end{proposition}

The terms `dense coding' and `superdense coding' are usually treated as synonyms in quantum theory. We use here the terminology `superdense coding' as follows.

\theoremstyle{definition}
\newtheorem*{SDC}{Superdense coding}
\begin{SDC}
\label{SDC}
A superdense coding (SDC) protocol is a \hyperref[DC]{dense coding} protocol implemented with a state $\phi\in\Omega_{AB}$ whose capacity satisfies $\chi_{\text{DC}}(\phi)>\chi_{\text{C}}(\Omega_{A})$.
\end{SDC}

We prove now that superdense coding requires some type of entangled states, hence, it is impossible in classical probabilistic theories.
\begin{proposition}
\label{theorem0}
Superdense coding is impossible in a GPT with joint state space $\Omega_{AB}=\Omega_A\otimes_{\text{min}}\Omega_B$.
\end{proposition}
\begin{proof}
Suppose that $\Omega_{AB}=\Omega_A\otimes_{\text{min}}\Omega_B$ and consider the dense coding protocols described above. The initial state $\phi\in\Omega_{AB}$ shared by Alice and Bob is separable $\phi=\sum_z q_z\omega_z\otimes\omega'_z$, where $\omega_z\in\Omega_A$, $\omega'_z\in\Omega_B$ and $\{q_z\}_{z}$ is a probability distribution. After Alice's operation $T_x$, the state is transformed into $\phi_x=\sum_z q_z\omega_{z,x}\otimes\omega'_z$, where $\omega_{z,x}=T_x\omega_z\in\Omega_{A}$. 
Bob's measurement is $M=\{E_y\in\mathcal{E}_{AB}\}_y$.
The probability of outcome $y$ given input $x$ is $p(y\vert x)=E_y \cdot \left(\sum_z q_z\omega_{z,x}\otimes\omega'_z\right)$.
Convexity of the mutual information implies that
$I(X:Y)\leq \sum_z q_z I_z(X:Y)$
where $I_z(X:Y)$ is the mutual information between $x$ and $y$ when the initial state is $\omega_{z,x}\otimes\omega'_z$. 
Hence we can take the initial state to be a product state $\phi=\omega\otimes\omega'$, and drop the label $z$.
We can write 
$p(y\vert x)=E_y\cdot (\omega_{x}\otimes\omega')= \omega_{x} \cdot e'_{y}$
where $\omega_x=T_x\omega$ and $e'_y= E_y \omega'$.
The set $\{e'_y\}_y$ constitutes a measurement on $\Omega_A$ (i.e. $e'_y \cdot  \omega \geq 0 ~\forall \omega\in\Omega_A$ and $\sum_y e'_y =u_A$, which are implied by the fact that
$0\leq E_y\cdot \phi$ for all $\phi\in\Omega_{AB}$ and $\sum_y E_y =u_{AB}$).
Therefore $p(y\vert x)=e'_y \cdot  \omega_x$ is a probability distribution that can be obtained by a measurement on $\Omega_A$. Therefore $\chi_{\text{DC}}(\Omega_{AB})\leq \chi_{\text{C}}(\Omega_A)$, hence, there cannot be superdense coding.
\end{proof}

Intuitively, one also expects that superdense coding requires entangled effects, that is, that superdense coding is impossible if $\mathcal{E}_{AB}=\mathcal{E}_A\otimes_{\text{min}}\mathcal{E}_B$. It is easy to show that superdense coding is impossible if Bob's measurement is a convex combination of product measurements:
\begin{proposition}
\label{theorem0.5}
Superdense coding is impossible if Bob's decoding measurement consists of effects of the form $E_{y_1y_2}=\sum_z q_z e_{y_1}^{(z)}\otimes f_{y_2}^{(z)}$ with $\{q_z\}_z$ a probability distribution, and $\{e_{y_1}^{(z)}\in\mathcal{E}_A\vert \sum_{y_1} e_{y_1}^{(z)}=u_A\}$ and $\{f_{y_2}^{(z)}\in\mathcal{E}_B\vert  \sum_{y_2} f_{y_2}^{(z)}=u_B\}$ being measurements on systems $A$ and $B$, respectively.
\end{proposition}
\begin{proof}
Let Bob's measurement correspond to effects
\begin{equation}
\label{eq:d2}
E_{y_1,y_2}=\sum_z q_z e_{y_1}^{(z)} \otimes f_{y_2}^{(z)}\ ,
\end{equation}
with  $\{e_{y_1}^{(z)}\in\mathcal{E}_A\vert \sum_{y_1}e_{y_1}^{(z)}=u_A\}_{y_1}$ and $\{f_{y_2}^{(z)}\in\mathcal{E}_B\vert\sum_{y_2}f_{y_2}^{(z)} =u_B\}_{y_2}$ being measurements on systems $A$ and $B$, respectively. We show that in this case there cannot be \hyperref[SDC]{superdense coding}, that is, we show that $I(X:Y_1,Y_2)\leq \chi_{\text{C}}(\Omega_A)$, where $I(X:Y)$ is the mutual information between $x$ and $y$.

Convexity of the mutual information implies that
$I(X:Y)\leq \sum_z q_z I_z(X:Y)$
where $I_z(X:Y)$ is the mutual information between $x$ and $y$ for 
the measurement $E^{(z)}_{y_1,y_2}= e_{y_1}^{(z)} \otimes f_{y_2}^{(z)}$. Hence, we can  drop the index $z$, and
 consider a product measurement $E_{y_1,y_2}= e_{y_1} \otimes f_{y_2}$.
The probability of outcome $y_1,y_2$ on state $\phi_x=T_x \phi$ is 
$p(y_1,y_2\vert x)=(e_{y_1} \otimes f_{y_2})\cdot\phi_x$. 
The no-signalling principle implies that $p(y_2\vert x)=p(y_2)$, hence,
$ p(y_1,y_2\vert x)=p(y_2) p(y_1\vert y_2, x)$. Therefore,
$I(X:Y_1,Y_2)= I(X:Y_2)+I(X:Y_1\vert Y_2)=I(X:Y_1\vert Y_2)$, since $Y_2$ is independent of $X$.

We have that $p(y_1 \vert y_2, x )= p(y_2\vert x)^{-1} p(y_1,y_2\vert x)= e_{y_1}\cdot \omega_{x,y_2}$, where $\omega_{x,y_2}\equiv\bigl((u_A\otimes f_{y_2})\cdot \phi_x\bigr)^{-1}\phi_xf_{y_2}$. This is because $p(y_2\vert x)=\sum_{y_1}p(y_1,y_2\vert x)=(u_A\otimes f_{y_2})\cdot \phi_x$ and $p(y_1,y_2\vert x)= (e_{y_1}\otimes f_{y_2})\cdot \phi_x = e_{y_1}\cdot (\phi_xf_{y_2})$. We also have that $\omega_{x,y_2}\in\Omega_A$ because i) for any effect $e\in\mathcal{E}_A$ it holds that $e\cdot \omega_{x,y_2} \geq 0$, due to the fact that $(e\otimes f_{y_2})\cdot \phi_x \geq 0$, and ii) $u_A\cdot \omega_{x,y_2} = 1$, as follows from the definition of $\omega_{x,y_2}$. Thus, we have $I(X:Y_1\vert Y_2)\leq \chi_{\text{C}}(\Omega_A)$.
\end{proof}

We note that if $\mathcal{E}_{AB}=\mathcal{E}_A\otimes_{\text{min}}\mathcal{E}_B$, the most general measurements do not consist only of convex combinations of product measurements as above, but also include measurements in which each effect is a convex combination of product effects:
$\{ E_y= \sum_z q_z^{(y)} e^{(z)}_{y}\otimes f^{(z)}_{y}\vert  \sum_y E_y =u_{AB}\}$, with $e_{y}^{(z)}\in\mathcal{E}_A$,  $f_{y}^{(z)}\in\mathcal{E}_B$ (but $\{e_{y_1}^{(z)}\}$ and $\{f_{y_2}^{(z)}\}$ do not necessarily need to define independent measurements, i.e., we do not need to have $\sum_{y_1}e_{y_1}^{(z)}=u_A$ and $\sum_{y_2}f_{y_2}^{(z)} =u_B$). This set of measurements include in particular those in which the choice of measurement implemented on one of the systems depends on the outcome obtained on the other system (what is sometimes named as ``wirings''). For example, by setting $y\equiv (y_1,y_2)$ and $q_z^{(y_1,y_2)}= 1$ if $z=y_1$ and $q_z^{(y_1,y_2)}= 0$ otherwise, we obtain effects of the form $E_{y_1,y_2}=e_{y_1}\otimes f_{y_2}^{(y_1)}$. Another example of a measurement in this more general set that cannot be viewed as a convex sum of product measurements is given in quantum theory by the phenomenon of \emph {quantum nonlocality without entanglement} \cite{BDFMRSSW99}. We leave it as an open question whether superdense coding is possible with this class of measurements. A step towards proving this generalization of Proposition \ref{theorem0.5} is the fact that superdense coding is impossible if $\Omega_{AB}$ is a state space in generalized nonsignalling theory \cite{B07}, also called boxworld \cite{SB12}, where all nonsignalling states are allowed.

\subsection{Hyperdense coding}

Quantum theory allows superdense coding, since e.g. for the maximally entangled state $\phi$ of two qubits, $\chi_{\text{DC}}(\phi)=2\chi_{\text{C}}(\Omega_{A})$ where $\Omega_A$ is the state space of single qubits \cite{sdc}. However, superdense coding is limited in quantum theory as $\chi_{\text{DC}}(\phi)\leq 2\chi_{\text{C}}(\Omega_A)$ for any $\Omega_A$ and $\phi\in \Omega_{AB}$ . We call \emph{hyperdense} a dense coding protocol overcoming this quantum limitation.
\theoremstyle{definition}
\newtheorem*{HDC}{Hyperdense coding}
\begin{HDC}
\label{HDC}
A hyperdense coding (HDC) protocol is a \hyperref[DC]{dense coding} protocol implemented with a state $\phi\in\Omega_{AB}$ whose capacity satisfies $\chi_{\text{DC}}(\phi)>2\chi_{\text{C}}(\Omega_{A})$.
\end{HDC}

Note that if $\chi_{\text{C}}(\Omega_{A})=\chi_{\text{C}}(\Omega_{B})$, for example if $\Omega_A=\Omega_B$, then hyperdense coding implies superadditive classical capacities: $\chi_{\text{C}}(\Omega_{AB})\geq \chi_{\text{DC}}(\Omega_{AB})>2\chi_{\text{C}}(\Omega_{A})=\chi_{\text{C}}(\Omega_{A})+\chi_{\text{C}}(\Omega_{B})$, as follows from Proposition \ref{theochichi} and the given definitions.

\begin{proposition}
\label{HDCQT}
Hyperdense coding is impossible in quantum theory.
\end{proposition}

\begin{proof}
We give two simple arguments within quantum theory. First, the convexity of the mutual information \cite{CTbook} and the convex decomposition of mixed states into pure states imply that the dense coding capacity $\chi_{\text{DC}}(\Omega_{AB})$ is achieved by pure states. Furthermore, in the case that the state $\phi\in\Omega_{AB}$ initially shared by Alice and Bob is pure and Alice's quantum system $A$ has Hilbert space dimension $d$, the Schmidt decomposition of pure states \cite{NielsenandChuangbook} implies that the dimension of the Hilbert space $\phi$ that is accessible by Alice's transformations on $A$ is no greater than $d^2$. Thus, we have that $\chi_{\text{DC}}(\Omega_{AB})\leq 2\log_2 d=2\chi_{\text{C}}(\Omega_{A})$.

For the second alternative argument, consider classical systems $C$ and $D$, initially uncorrelated to Alice's and Bob's joint quantum system $AB$, that record Alice's preparation $x$ and Bob's measurement outcome $y$, respectively. The principle of quantum information causality \cite{QIC} states that the quantum mutual information between system $C$ and the joint system $ABD$, after Alice's transmission of a qudit $A$ satisfies the bound $I_{\text{Q}}(C:ABD)\leq 2\log_2 d$, independently of how big the Hilbert space dimension of Bob's system $B$ might be. It is easy to see from the data-processing inequality \cite{NielsenandChuangbook} that $I(X:Y)\leq I_{\text{Q}}(C:ABD)$, from which follows that $\chi_{\text{DC}}(\Omega_{AB})\leq 2\log_2 d=2\chi_{\text{C}}(\Omega_{A})$.
\end{proof}

\subsection{Hyperdense coding in hypersphere theories}
\label{HDCsec}
We now introduce a dense coding protocol in hypersphere theories of dimension $n=2^N-1$ which is hyperdense for $N>2$. 
Alice and Bob initially share the entangled state $\phi_{\bold{0}}=I\in\mathcal{S}_N\subset\Omega_{AB}$, which is the identity matrix in $\mathbb{R}^{2^N}\otimes\mathbb{R}^{2^N}$, as defined by (\ref{eq:43.310}).
We consider messages $x\in\{0,1\}^N$. To encode $x$,
Alice applies the local transformation $T_x\in\mathcal{T}_N$ given by (\ref{eq:43.51}), which encodes the message $x\in\{0,1\}^N$. From (\ref{eq:43.51}) and (\ref{eq:43.52}), the joint state transforms into $T_x\phi_{\bold{0}}=\phi_x\in\mathcal{S}_N$. Alice sends the system $A$ to Bob. Bob applies the measurement $M=\{E_y\}_{y\in\{0,1\}^N}$ given by (\ref{eq:43.M}) on the joint system $AB$. The probability that Bob obtains outcome $y$ when Alice encodes the message $x$ is $p(y\vert x)=E_y\cdot\phi_x=\delta_{y,x}$. Thus, Bob decodes Alice's $N$ bit message perfectly. 
This shows that 
\begin{equation}
\chi_{\text{DC}}(\Omega_{AB})\geq N\ .
\label{chDCAA}
\end{equation}
Since the classical capacity of individual systems in HSTs is 1 bit, this provides an example of hyperdense coding for $N>2$.
By taking $N\rightarrow\infty$, Bob can learn an arbitrarily large amount of information by receiving a system that locally can only encode up to one bit.

The above dense coding protocol can be turned into a classical communication protocol in which Alice encodes message $x$ in the state $\phi_x$ of two HST systems, and Bob decodes the message perfectly using measurement $M=\{E_y\}_{y\in\{0,1\}^N}$. Hence (see Proposition \ref{theochichi})
\begin{equation}
2N\geq \chi_{\text{C}}(\Omega_{AB})\geq \chi_{\text{DC}}(\Omega_{AB})\geq N\ .
\label{chCCAA}
\end{equation}
Even though each system has capacity of $1$ bit, together they have capacity
 of at least $N$ bits, with arbitrary $N\in\mathbb{N}$. Thus, in the theories defined above, the classical capacities are superadditive.

The upper bound of $2N$ (the left hand side of (\ref{chCCAA})) follows from the fact that the classical capacity of any GPT is bounded by the $\log_2$ of the dimension of the state space \cite{FMPT13}, which in the case of entangled states for a pair of HST systems, as defined above, is $2N$. Thus, the gap we exhibit between single  system classical capacity and two system classical capacity is close to optimal.

\subsection{Teleportation and entanglement swapping in hypersphere theories}
Quantum teleportation \cite{teleportation} and entanglement swapping \cite{ZZHE93} are fundamental protocols of quantum information theory. We show how they can also be realized in the context of hypersphere theories considered here. We notice that teleportation and entanglement swapping in GPTs have been studied before \cite{B07,BBLW08,SB12}. In particular, it has been shown that boxworld, the theory admitting all nonsignalling correlations, does not have teleportation or entanglement swapping \cite{B07,SB12}. General conditions on GPTs to support teleportation are given in \cite{BBLW08}.

We consider the following protocol. Alice has a system $A'$ in a pure state $\omega_{a}\in\Omega_{A'}$ that she wants to teleport to Bob's location. To do so, Alice and Bob initially share an entangled state $\phi_{\bold{0}}\in\mathcal{S}_N\subset\Omega_{AB}$ given by (\ref{eq:43.310}) in the systems $A$ and $B$, at Alice's and Bob's locations, respectively. Alice applies the measurement defined by the entangled effects $E_x\in\mathcal{F}_N\subset\mathcal{E}_{AB}$ given by (\ref{eq:43.M}) on her joint system $A'A$ and obtains the outcome $x$ with probability $p_x=2^{-N}$, for $x\in\{0,1\}^N$. Alice sends Bob her outcome. Bob applies the correction operation $T_x\in\mathcal{T}_N$ given by (\ref{eq:43.51}) on his system $B$ and obtains, as we show below, the teleported state $\omega_{a}$ on $B$ with unit probability. From the linearity of the theory, this protocol works too if Alice's input state $\omega_{a}$ is mixed. In particular, the system $A'$ can be in an entangled state $\varphi\in\Omega_{A'C}$ with another system $C$, leading to entanglement swapping.  Furthermore, we show that for theories with bipartite state spaces $\Omega_{AB}$ described above, our protocols for teleportation and entanglement swapping are optimal in the amount of classical information sent by Alice, that is, teleportation or entanglement swapping cannot be achieved perfectly with a classical channel of capacity smaller than $N$ bits.

Let us now show that these protocols for teleportation and entanglement swapping work as claimed.  To do so, we first define a consistent state space for a tripartite system: $\Omega_{A'AB}\equiv\text{convex hull }\{\Omega_{A'}\otimes_{\text{min}}\Omega_{AB},\Omega_{A}\otimes_{\text{min}}\Omega_{A'B},\Omega_{B}\otimes_{\text{min}}\Omega_{A'A}\}$, where $\Omega_{AB}\equiv\text{convex hull }\{\Omega_{A}\otimes_{\text{min}}\Omega_{B},\mathcal{S}_N\}$ and $\mathcal{S}_N\equiv\{\phi_\mu\}_{\mu\in\{0,1\}^N}$, as before. We define the four party state space in a similar way, $\Omega_{A'ABC}\equiv\text{convex hull }\{\Omega_{A'AB}\otimes_{\text{min}}\Omega_{C},\Omega_{A'AC}\otimes_{\text{min}}\Omega_B,\Omega_{ABC}\otimes_{\text{min}}\Omega_{A'},\Omega_{A'BC}\otimes_{\text{min}}\Omega_A\}$. We notice that the three and four party state spaces are symmetric under any permutation of the systems, all product states are included, and the set of bipartite entangled states $\mathcal{S}_N$ is included in any bipartition. It follows that the set of global effects $\mathcal{E}_{A'ABC}$ includes all the product effects $E\otimes E'$, where $E'\in\mathcal{E}_{BC}$ and $E\in\mathcal{E}_{A'A}$, in particular for the bipartite entangled effects $E,E'\in \mathcal{F}_N$, for any permutation of the systems $A'$, $A$, $B$ and $C$.

The state $\omega_{a}$ is successfully teleported to system $B$ if, for any measurement on $B$, the outcome probabilities are those predicted by $\omega_a$. In our notation, this means that
\begin{eqnarray}
\label{eq:071}
p_{\text{tel}}(y\vert x)&\equiv& \frac{1}{p_x}\bigl( (E_x)_{A'A}\otimes (e_{y})_{B}\bigr)\cdot\bigl((\omega_{a})_{A'}\otimes(\phi T_x^{\text{t}})_{AB}\bigr)\nonumber\\
&=&e_{y}\cdot \omega_{a},
\end{eqnarray}
for any effect $e_y\in\mathcal{E}_B$, where the label `$\text{tel}$' denotes that these are Bob's outcome probabilities after performing the teleportation protocol given above.

We show that the probability that Alice and Bob obtain respective outcomes  corresponding to the effects $E_x\in\mathcal{E}_{A'A}$ and $e_y\in\mathcal{E}_B$ is 
\begin{eqnarray}
\label{eq:92}
p_{\text{tel}}(x,y)&\equiv& \bigl( (E_x)_{A'A}\otimes (e_{y})_{B}\bigr)\cdot\bigl((\omega_{a})_{A'}\otimes(\phi_{\bold{0}} T_x^{\text{t}})_{AB}\bigr)\nonumber\\
&=&2^{-N}e_{y}\cdot \omega_{a}.
\end{eqnarray}
Since the probability that Alice obtains outcome $x$ is $p_x\equiv\sum_{y}p_{\text{tel}}(x,y)$ and since $\sum_{y}e_y=u$, we have from (\ref{eq:92}) that $p_x=2^{-N}$. Thus, we have $p_{\text{tel}}(y\vert x)\equiv \frac{1}{p_x}p_{\text{tel}}(x,y)=e_{y}\cdot \omega_{a}$, as given by (\ref{eq:071}).

We show (\ref{eq:92}). Since the states $\phi_x$ and the local transformations $T_x$ are diagonal matrices, as given by (\ref{eq:43.310}) and (\ref{eq:43.51}), we have from (\ref{eq:43.51}) and (\ref{eq:43.52}) that $\phi_{\bold{0}}T_x^{\text{t}}=\phi_x$. Using the indices $i$, $j$ and $k$ for the systems $A'$, $A$ and $B$, respectively, we obtain from (\ref{eq:92}) that
\begin{eqnarray} 
\label{eq:73}
p_{\text{tel}}(x,y)&\equiv&\sum_{i,j,k\in\{0,1\}^N} (E_x)_{i,j}(e_y)_k (\omega_{a})_i (\phi_x)_{j,k}\nonumber\\
&=&2^{-N}\sum_{i,j,k\in\{0,1\}^N} (\phi_x)_{i,j}(e_y)_k (\omega_{a})_i (\phi_x)_{j,k}\nonumber\\
&=&2^{-N}\sum_{i,j,k\in\{0,1\}^N} (\phi_x)_{i,i}\delta_{i,j}(e_y)_k (\omega_{a})_i (\phi_x)_{j,j}\delta_{j,k}\nonumber\\
&=&2^{-N}\sum_{k\in\{0,1\}^N} (\phi_x)_{k,k}(e_y)_k (\omega_{a})_k (\phi_x)_{k,k}\nonumber\\
&=&2^{-N}e_{y}\cdot\omega_{a},
\end{eqnarray}
where in the second line we used the definition for the effects $E_x$ given by (\ref{eq:43.FN}), in the third line we used that the states $\phi_x$ are diagonal matrices, as given by (\ref{eq:43.310}), and in the last line we used that the entries of $\phi_x$ are $1$ or $-1$, as given by (\ref{eq:43.0}) and (\ref{eq:43.310}).

The introduced teleportation protocol can easily be extended to entanglement swapping. Let Alice's system $A'$ be initially in an entangled state $\varphi\in\Omega_{A'C}$ with a system $C$, held by Charlie. As in the teleportation protocol, Alice and Bob share the entangled state $\phi_{\bold{0}}\in\Omega_{AB}$ in systems $A$ and $B$. Then, Alice and Bob perform the teleportation protocol described on the tripartite system $A'AB$. 

The entangled state $\varphi$ is successfully swapped from systems $A'C$ to systems $BC$ if the outcome probabilities of any measurement on the joint system $BC$ are those predicted by the state $\varphi$. The probability that Alice obtains the outcome $E_x\in\mathcal{E}_{A'A}$ on the system $A'A$ and Bob obtains (by collaborating with Charlie) the outcome $E'_y\in\mathcal{E}_{BC}$ on the system $BC$, after he applies the correction $T_x$ on $B$, is
\begin{equation}
\label{eq:74}
p_{\text{swap}}(x,y)\equiv\bigl( (E_x)_{A'A}\otimes (E'_{y})_{BC}\bigr)\cdot\bigl((\varphi)_{A'C}\otimes(\phi_{\bold{0}} T_x^{\text{t}})_{AB}\bigr).
\end{equation}
Entanglement swapping succeeds if $p_{\text{swap}}(y\vert x)\equiv\frac{1}{p_x}p_{\text{swap}}(x,y)=E_y'\cdot \varphi$ for all  $\varphi\in\Omega_{A'C}$ and $E_{y}'\in\mathcal{E}_{BC}$, where $p_x\equiv\sum_{y}p_{\text{swap}}(x,y)>0$. This follows straightforwardly as in (\ref{eq:73}), with $p_x=2^{-N}$.

Finally, we show that for theories with bipartite state spaces $\Omega_{AB}$, teleportation or entanglement swapping cannot be perfectly implemented with a classical channel of capacity of less than $N$ bits. We show this property for entanglement swapping. Since the linearity of the theory allows us to extend a protocol for teleportation to a protocol for entanglement swapping, as we did above, this property holds for teleportation too. We follow the argument given in \cite{teleportation}. Assume that Alice and Bob implement the dense coding protocol of section \ref{HDCsec} with an initial state $\phi_{\bold{0}}$ in systems $A$ and $B$ at Alice's and Bob's locations, respectively. The state is  transformed into $\phi_x$ after Alice's operation, encoding the $N$ bit message $x$. Suppose now that Alice does not send Bob the system $A$, but that Alice and Bob apply some entanglement swapping protocol, using some resource state in systems $A'$, held by Alice, and $B'$, held by Bob, so that at the end of the protocol Bob has the state $\phi_x$ in his systems $B'B$. Then, Bob completes the dense coding protocol and learns the message $x$. Assume that such a protocol for entanglement swapping requires Alice to encode Bob's correction operation in a message $y$ via a classical channel with capacity of $M<N$ bits. Now suppose that Alice does not actually send $y$ to Bob, but that Bob instead tries to guess the value of $y$ to complete the protocol without any communication. Then, Bob will guess $y$ correctly with probability at most $2^{-M}$, in which case he completes entanglement swapping successfully and then learns $x$. Thus, Bob obtains Alice's $N$ bit message $x$ with probability $2^{-M}>2^{-N}$ without any communication from Alice to Bob. This means that, from the set of $2^N$ messages of $N$ bits that Alice can encode, there exists at least a pair of messages that Bob can distinguish with a nonzero probability without receiving any physical system from Alice, which violates the no-signalling principle. It follows that a classical channel of capacity $N$ bits is necessary to complete entanglement swapping or teleportation.

\section{Discussion}
Here we have introduced hyperdense coding: superdense coding in which more than two bits of information can be communicated by transmission of a system that locally encodes at most one bit. We have presented dense coding protocols in the context of hypersphere theories, in which single systems are described by an Euclidean ball of dimension $n$. Our protocols are hyperdense when $n>3$.

It is well known that if one imposes a sufficient set of axioms to GPTs then one recovers classical or quantum theory \cite{H01,DB09,Chiribella11,MM11,MMAP13}. Therefore, the theories we have introduced must violate at least one of these axioms. The theories we introduced violate continuous reversibility. Continuous reversibility is an important physical condition that has been imposed to the framework of GPTs in several derivations of finite dimensional quantum theory \cite{H01,DB09,MM11,MMAP13}. Continuous reversibility states that for every pair of pure states there exists a continuous reversible transformation that transforms one into the other \cite {H01}.

The theories that we introduced not only violate \emph{continuous reversibility}, but already simply \emph{reversibility}. Indeed, the set of transformations acting on a single system is discrete, see Eq. (\ref{eq:43.51}), while the set of local pure states is continuous, given by a hypersphere. Thus, not every pair of pure states describing a single system is connected by a transformation. 

Note further that if we modify the group of transformations $\hat{\mathcal{T}}$ for the theories given above to be $\hat{\mathcal{T}}=\text{SO}(n)$ so that continuous reversibility is satisfied (at least for local systems), then the space of entangled states must be modified and the hyperdense coding protocols we introduced no longer work. In fact, we obtain in this case that less than one bit is communicated if $n\neq 3$, while $2$ bits are communicated if $n=3$, which corresponds to the Bloch ball (see details in \hyperref[app:continuoushdc]{Appendix \ref{app:continuoushdc}}). More generally, we have shown, under a variety of technical conditions, that for arbitrary hypersphere theories, hyperdense coding is impossible if one imposes continuous reversibility on local systems (see \hyperref[app:theorem]{Appendix \ref{app:theorem}}). 
%We have shown that, for a state space $\Omega_{AB}$ with $\Omega_A\simeq\Omega_B\simeq\Omega$ being hypersphere systems of dimension $n\geq 2$ satisfying tomographic locality and local continuous reversibility, there is not superadditivity of classical capacities or hyperdense coding under a variety of technical conditions (see Proposition \ref{theorem2} in the \hyperref[app:theorem]{appendix}). In particular, there is not superadditivity, and hence no hyperdense coding,  if $n$ is odd for states that are locally completely mixed, or if $n$ is even for arbitrary states. Moreover, if the local states are completely mixed, superdense coding cannot be as optimal as in the quantum case for $n\neq 3$, and superdense coding is impossible for even $n$ \footnote{We showed these results if the set of allowed local transformations is $\hat{\mathcal{T}}=\text{SO}(n)$. In general, local continuous reversibility requires $\hat{\mathcal{T}}$ to be transitive on the unit sphere in $\mathbb{R}^n$. For odd $n\neq 7$ the only possibility is $\hat{\mathcal{T}}=\text{SO}(n)$ \cite{MMPA14}.}. Intuitions borrowed from quantum theory suggest that if the local states are completely mixed then the bipartite states are maximally entangled and hence that superdense coding should be optimized for this class of states, thus suggesting that the mentioned results for superdense coding should hold for arbitrary states in $\Omega_{AB}$.

In general, continuous reversibility not only applies to single systems with state spaces $\Omega_A$ or $\Omega_B$, but also to composite systems in $\Omega_{AB}$, i.e., for  every pair of pure states in $\Omega_{AB}$ there should exist a continuous reversible transformation that transforms one into the other. In line with the above remarks, note that imposing continuous reversibility for composite systems implies that superdense coding is possible in HSTs only if $n=3$, which corresponds to the Bloch ball. This follows from the result of \cite{MMAP13} showing that for a pair of HST systems satisfying continuous reversibility, entangled states only exist for the case $n=3$. While our theories must thus obviously violate continuous reversibility, note that it is an open question whether they satisfy the weaker requirement that there merely exists a transformation mapping one of the product state to one of the entangled states used in our protocol. 

In addition to continuous reversibility, tomographic locality is another physical conditions that has been used to derive finite dimensional quantum theory in the framework of GPTs \cite{H01,DB09,MM11,MMAP13}. Tomographic locality is a physical property stating that the description of entangled states is completely fixed by the outcome probabilities of local measurements \cite{H01,B07}. Violation of tomographic locality means that there are global degrees of freedom that are inaccessible to local observers. The theories that we have discussed satisfy tomographic locality. It is not difficult to build theories that exhibit hyperdense coding and that satisfy local continuous reversibility, but that violate tomographic locality (see \hyperref[app:lcrtlhdc]{Appendix \ref{app:lcrtlhdc}}).

We remark that while hyperdense coding must necessarily entail a violation of one of the physical axioms defining quantum theory, such as tomographic locality or continuous reversibility, it is not necessarily the case that any theory violating one of these axioms always allows for the existence of hyperdense coding, even if it includes entangled states. Intuitively, one can construct theories in which the entangled states are very weekly entangled
and thus do not achieve hyperdense coding. We give examples of this type in \hyperref[lastapp]{Appendix \ref{lastapp}}, where we construct theories whose local systems are described by HSTs and whose bipartite state spaces include a set of nontrivial entangled states. These theories satisfy tomographic locality and violate continuous reversibility, for local and global systems. We show that some of these theories do not have hyperdense coding, while others do not even have superdense coding.

On a more general level, our hyperdense coding protocols have highly unphysical consequences. Indeed, they imply superadditive capacities: the classical capacity of the bipartite system $AB$ is greater than the sum of the local capacities of $A$ and $B$. A breakdown of additivity suggests that in such theories one cannot define a unit of information, and hence that the whole framework of information theory breaks down. 
(We note that in quantum theory there is a weak breakdown of additivity, namely the capacity of a specific channel can be superadditive \cite{H09}. This has much less dramatic consequences than the superadditivity discussed here).

Another possible consequence is related to thermodynamics. Indeed, the entropy of a state $\omega$ can be defined, as briefly suggested in \cite{BBCLSSWW10}, as the maximum classical capacity of any encoding/decoding protocol in which  Alice sends Bob the state $\omega=\sum_x p_x \omega_x$ on average. This definition coincides with the Shannon and von Neumann entropies in the classical and quantum cases. It differs however from other definitions of entropy in GPTs \cite{BBCLSSWW10,SW10}. Our hyperdense coding protocols show that using the definition based on classical capacity, entropy can be superadditive in GPTs. This suggests that statistical mechanics could not be applied to these theories, and that they would not have  a macroscopic limit.

For these reasons, it would be  interesting to investigate whether additivity of the classical capacity should be taken as a basic physical condition for any reasonable theory. We have suggested that it is related to other properties, including tomographic locality and continuous reversibility.  What the detailed relation is, and whether additivity of classical capacities can replace other physical conditions previously explored to distinguish quantum theory, is an open question. 

%\begin{acknowledgments} {\bf Acknowledgments}
\section*{Acknowledgments}
We acknowledge financial support from the European Union under the project QALGO, from the F.R.S.-FNRS under the project DIQIP, from the InterUniversity Attraction Poles through project Photonics@be, and from the Brussels-Capital Region through a BB2B grant.
S. P. is a Research Associate of the Fonds de la Recherche Scientifique F.R.S.-FNRS (Belgium).
%\end{acknowledgments}

\appendix
\section{Locally continuous hypersphere theories}
\label{app:continuoushdc}

We define precisely the notion of continuous reversibility for local systems.

\theoremstyle{definition}
\newtheorem*{LCT}{Local Continuous Reversibility (LCR)}
\begin{LCT}
\label{LCT}
For a bipartite system $AB$, any pair of pure states for the local system $A$, or $B$, is connected by a continuous reversible transformation.
\end{LCT}

We investigate the implications of imposing \hyperref[LCT]{local continuous reversibility} on the dense coding protocol given in the main text. In order to ensure the \hyperref[consistency]{consistency} condition, we modify the space of states and effects in a minimal way, in terms of two parameters $\lambda$ and $\tau$. 
Recall that $n=2^N-1$ is the dimension of the local HST systems, with $N\in\mathbb{N}$.
We show that if we impose satisfaction of local continuous reversibility then our dense coding protocols communicate two bits only for $N=2$ and no more than one bit for $N\neq 2$, hence, they are superdense only for $N=2$, which corresponds to the Bloch sphere.

We consider entangled states and effects similar to the ones in the main text, but modified in terms of two parameters $\lambda$ and $\tau$:
 \begin{equation}
\label{eq:47}
\phi_\mu^{(\lambda)}=
 \begin{pmatrix}
1 & \bold{0} \\
 \bold{0}  & \lambda\hat{T}_\mu
 \end{pmatrix},
\qquad
E_\mu^{(\tau)} \equiv 2^{-N} 
\begin{pmatrix}
1 & \bold{0} \\
 \bold{0}  & \tau\hat{T}_\mu
 \end{pmatrix},
\end{equation}
where $\hat{T}_\mu$ is defined as in Eq. (\ref{eq:43.53}) of the main text, for $\mu\in\{0,1\}^N$ and some $\lambda,\tau\in[-1,1]$ that we specify below.  Notice that the case $\lambda=\tau=1$ corresponds to the states and effects given by Eqs. (\ref{eq:43.310}) and (\ref{eq:43.FN}) of the main text. Equations (\ref{eq:43.2}), (\ref{eq:43.310}), (\ref{eq:43.51}) and (\ref{eq:43.53}) of the main text, and Eq. (\ref{eq:47}) imply that $\sum_{\mu\in\{0,1\}^N} E_\mu^{(\tau)} =\text{diag} (1,\bold{0})$, which is the unit effect $u_{AB}$, for any value of $\tau$. 

We impose that local continuous reversibility must be satisfied. We must have that $\mathcal{T}_A=\mathcal{T}_B=\mathcal{T}$, where $\mathcal{T}$ is defined as the set of transformations of the form
\begin{equation}
\label{eq:48}
T=
\begin{pmatrix}
1 & \bold{0} \\
 \bold{0}  & \hat{T}
 \end{pmatrix},
\end{equation}
with $\hat{T}\in\hat{\mathcal{T}}$ and $\hat{\mathcal{T}}$ being a group of continuous reversible  transformation that is transitive on the sphere in $\mathbb{R}^{2^N-1}$ \cite{MMAP13}. It follows that $\hat{\mathcal{T}}$ is a subgroup of $\text{SO}(2^N-1)$ \cite{MMPA14}. Since $\text{SO}(1)=\{1\}$, there are no continuous reversible transformations for the case $N=1$. In fact, the case $N=1$ corresponds to the classical bit, whose transformations are discrete and correspond to $\hat{\mathcal{T}}=\text{O}(1)=\{1,-1\}$. Then, there is no superdense coding for $N=1$ since superdense coding does not exist in classical probabilistic theory. Thus, in the following, we  only consider $N\geq 2$. For $N\neq 3$, the only possible group of continuous reversible transformations that is transitive on the sphere is $\hat{\mathcal{T}}=\text{SO}(2^N-1)$, while for $N=3$, $\hat{\mathcal{T}}$ can be either $\text{SO}(7)$ or G$_2$ \cite{MMPA14}. We consider here that $\hat{\mathcal{T}}=\text{SO}(2^N-1)$ for all $N\in\mathbb{N}/\{1\}$.

For a given $\lambda$, in order to have that the allowed local transformations satisfy the \hyperref[consistency]{consistency} condition, it must be that $T\phi_\mu^{(\lambda)} T'^{\text{t}}\in\Omega_{AB}$, for all $T,T'\in\mathcal{T}$ and $\mu\in\{0,1\}^N$. We thus define $\Omega_{AB}\equiv \text{convex hull}\{\Omega_A\otimes_{\text{min}}\Omega_B,\Lambda^{(\lambda)}\}$, where $\Lambda^{(\lambda)}\equiv \{T\phi_{\bold{0}}^{(\lambda)} T'^{\text{t}} \vert T,T'\in\mathcal{T}\}$ is a continuous set of entangled states, differently to the discrete set $\mathcal{S}_N$ defined by the states (\ref{eq:43.310}) of the main text. We notice that $\phi_\mu^{(\lambda)}\in\Lambda^{(\lambda)}$ for all $\mu\in\{0,1\}^N$. Since $\lvert \lambda \rvert \leq 1$, following Eq. (\ref{eq:43.31}) of the main text, it is straightforward to show that $\Omega_{AB}\subseteq\Omega_A\otimes_{\text{max}}\Omega_B$.

We show that $\{E_\mu^{(\tau)}\}_{\mu\in\{0,1\}^N}\subset\mathcal{E}_{AB}$ requires
\begin{equation}
\label{eq:49}
-\bigl(2^N-1\bigr)^{-1}\leq \lambda\tau\leq \bigl(2^N-3\bigr)^{-1}.
\end{equation}
Consider the states $\phi_{\bold{0}}^{(\lambda)}$ and $\phi^{(\lambda)}=\phi_{\bold{0}}^{(\lambda)} T'$, where $T'=\text{diag}(1,1,-1,-1\ldots,-1)\in\mathbb{R}^{2^N}\otimes\mathbb{R}^{2^N}$. We see that $T'$ is of the form (\ref{eq:48}) with $\hat{T'}\in\text{SO}(2^N-1)$. We have that
\begin{eqnarray}
\label{eq:49.1}
E^{(\tau)}_{\bold{0}}\cdot \phi_{\bold{0}}^{(\lambda)}&=& 2^{-N}\Bigl(1+\bigl(2^N-1\bigr)\lambda\tau \Bigr),\nonumber\\
E^{(\tau)}_{\bold{0}}\cdot \phi^{(\lambda)}&=& 2^{-N}\Bigl(1-\bigl(2^N-3\bigr)\lambda\tau \Bigr).
\end{eqnarray}
Since $\phi^{(\lambda)}_{\bold{0}}, \phi^{(\lambda)}\in\Omega_{AB}$, the condition $E^{(\tau)}_{\bold{0}}\in\mathcal{E}_{AB}$ requires $0\leq E^{(\tau)}_{\bold{0}}\cdot \phi^{(\lambda)}_{\bold{0}}\leq 1$ and $0 \leq E^{(\tau)}_{\bold{0}}\cdot \phi^{(\lambda)}\leq 1$, which from (\ref{eq:49.1}) implies (\ref{eq:49}).

From Eqs. (\ref{eq:43.3}), (\ref{eq:43.310}), (\ref{eq:43.51}) and (\ref{eq:43.53}) of the main text, and Eq. (\ref{eq:47}), Bob's outcome probabilities are
\begin{eqnarray}
\label{eq:50}
p(y\vert x)&=&E^{(\tau)}_y\cdot \phi^{(\lambda)}_x\nonumber\\
&=& 2^{-N}(\tau\lambda d_y\cdot d_x+1-\lambda\tau)\nonumber\\ 
&=& 2^{-N}(2^N\tau\lambda \delta_{y,x}+1-\lambda\tau)\nonumber\\
&=&\lambda\tau\delta_{y,x}+2^{-N}(1-\lambda\tau).
\end{eqnarray}
Satisfaction of (\ref{eq:49}) implies that $0\leq p(y\vert x)\leq 1$. The probability that Bob learns the correct message is
\begin{eqnarray}
\label{eq:51}
p(y=x\vert x)&=&P_{\lambda,\tau,N},\nonumber\\
&\leq&Q_N,
\end{eqnarray}
where
\begin{eqnarray}
\label{eq:51.1}
P_{\lambda,\tau,N}&\equiv&2^{-N}\bigl(1+(2^N-1)\lambda\tau\bigr),\\
\label{eq:51.2}
Q_N&\equiv& 2^{-N+1}(2^N-2)(2^N-3)^{-1}.
\end{eqnarray}
If $y\neq x$, we have $p(y \vert x)= (2^N-1)^{-1}(1-P_{\lambda,\tau,N})$. Given that $p_x=2^{-N}$ for $x\in\{0,1\}^N$, it is straightforward to obtain that
\begin{equation}
\label{eq:52}
I(X:Y)=N-H(P_{\lambda,\tau,N}),
\end{equation}
where 
$H(p)\equiv h(p)+(1-p)\log_2{(2^N-1)}$ is the Shannon entropy of a random variable $Z$ taking values $Z=0$ with probability $p$ and $Z=j$ with probability $(2^N-1)^{-1}(1-p)$, for $j=1,2,\ldots,2^N-1$, and $h(p)\equiv-p\log_2{p}-(1-p)\log_2(1-p)$ is the binary entropy. The maximum value of $I(X:Y)$ is achieved by minimizing $H(P_{\lambda,\tau,N})$, which is obtained for $P_{\lambda,\tau,N}=Q_N$. We thus consider $\lambda\tau=(2^N-3)^{-1}$, in which case we have $P_{\lambda,\tau,N}=Q_N$. Thus, from (\ref{eq:52}), the maximum value of $I(X:Y)$ for the considered protocols is
\begin{equation}
\label{eq:53}
I(X:Y)=N-H(Q_N).
\end{equation}
The achieved values of $I(X:Y)$ for $N=2,3,4,5$ are $2,0.15,0.05,0.02$, respectively. In general, we see from (\ref{eq:51.2}) and (\ref{eq:53}) that $I(X:Y)$ decreases with $N$. Thus, our protocols achieve superdense coding only for $N=2$.

\section{Communication capacities of locally continuous hypersphere theories}
\label{app:theorem}

We prove the following  results stating the impossibility of superadditive capacities, hyperdense and superdense coding for particular classes of bipartite HST systems.

\begin{proposition}
\label{theorem2}
Let $\Omega_{AB}$ be a bipartite state space satisfying \hyperref[NS]{the no-signalling principle}, \hyperref[TL]{tomographic locality} and \hyperref[LCT]{local continuous reversibility} with $\Omega_A\simeq\Omega_B\simeq\Omega$ a hypersphere system of dimension $n\geq 2$. Consider states $\phi_x\in\Omega_{AB}$ that encode a message $x$ with probability $p_x$, whose local states are $\omega_{a_x}=\bigl(\begin{smallmatrix}
1\\ a_x
\end{smallmatrix} \bigr)\in\Omega_A$ and $\omega_{b_x}=\bigl(\begin{smallmatrix}
1\\ b_x
\end{smallmatrix} \bigr)\in\Omega_B$, with $a_x,b_x\in\mathbb{R}^{n}$ and $\lVert a_x\rVert\leq 1$, $\lVert b_x\rVert\leq 1$. Let $y$ be the outcome of a measurement on the state $\phi_x$ corresponding to the effect $E_y\in\mathcal{E}_{AB}$ and let $I(X:Y)$ be the mutual information between $x$ and $y$. Let $\mathcal{T}$ be the set of allowed local transformations on $\Omega_A$ and $\Omega_B$. In general, their elements are of the form
\begin{equation}
\label{eq:a0}
T\equiv 
 \begin{pmatrix}
1 & \bold{0} \\
 \bold{0}  & \hat{T}
 \end{pmatrix},
\end{equation}
where the matrices $\hat{T}$ form the group $\hat{\mathcal{T}}$. The following properties hold:

i) for states with $a_x=\bold{0}$ or $b_x=\bold{0}$, we have $I(X:Y)\leq 2$ for odd $n\neq 7$, with equality achieved only for $n=3$;

ii) if $\hat{\mathcal{T}}=\text{SO}(n)$ and $n$ is odd then $I(X:Y)\leq 2$ for $a_x =\bold{0}$ or $b_x=\bold{0}$, with equality achieved only for $n=3$;

iii) if $\hat{\mathcal{T}}=\text{SO}(n)$ and $n$ is even then $I(X:Y)\leq 2$ for arbitrary $\phi_x\in\Omega_{AB}$, and $I(X:Y)\leq 1$ for states with $a_x=\bold{0}$ or $b_x=\bold{0}$;

iv) if $\hat{\mathcal{T}}=\text{O}(n)$, which includes continuous and discontinuous transformations, we have for any $n$ that, $I(X:Y)\leq 2$ for arbitrary $\phi_x\in\Omega_{AB}$ and $I(X:Y)\leq 1$ for states with $a_x=\bold{0}$ or $b_x=\bold{0}$;

v) for states that are obtained by local transformations $\phi_x=T_x\phi$, we have $I(X:Y)< 2$ if $\hat{\mathcal{T}}=\text{SO}(n)$ with even $n$, or if $\hat{\mathcal{T}}=\text{O}(n)$ with arbitrary $n$. 

Recalling that the classical capacity of a single system in hypersphere theories is 1 bit, no protocol using the above families of states can achieve superadditivity of classical capacities, hyperdense coding or superdense coding, for the respective cases.
\end{proposition}

The following lemmas are used in the proof of Proposition \ref{theorem2}, given above. They provide important constraints on the bipartite states and effects in HSTs for arbitrary state spaces $\Omega_{AB}$ limited only by \hyperref[NS]{NS}, \hyperref[TL]{TL} and the restriction that $\Omega_A\simeq\Omega_B\simeq\Omega$ is a HST system of dimension $n$.

\begin{lemma}
\label{lemma2}
Let $\Omega_{AB}$ be a bipartite state space that satisfies \hyperref[NS]{no-signalling} and \hyperref[TL]{tomographic locality}, with $\Omega_A\simeq\Omega_B\simeq\Omega$ being the state space of a HST system of dimension $n$. Any state $\phi\in\Omega_{AB}$ can be expressed as a matrix
\begin{equation}
\label{eq:11}
\phi\equiv 
 \begin{pmatrix}
1 & b^{\text{t}} \\
 a  & C
 \end{pmatrix},
\end{equation}
where $a,b\in\mathbb{R}^n$, $C\in\mathbb{R}^n\otimes\mathbb{R}^n$, $\lVert a\rVert \leq 1$, $\lVert b \rVert \leq 1$ and $\lVert c_k \rVert \leq 1$, with $c_k$ being the $k$th column vector of $C$, for $k=1,2,\ldots,n$.
\end{lemma}

\begin{proof}
The principles of \hyperref[NS]{no-signalling} and \hyperref[TL]{tomographic locality} imply that $\phi\in\Omega_{AB}$ can be expressed as the matrix (\ref{eq:11}), with $C\in\mathbb{R}^n\otimes\mathbb{R}^n$ and $a,b\in\mathbb{R}^n$, as given by Eq. (\ref{eq:011}) of the main text, and that $\Omega_A\otimes_{\text{min}}\Omega_B\subseteq\Omega_{AB}\subseteq\Omega_A\otimes_{\text{max}}\Omega_B$ \cite{BBLW07}. We show below that for HSTs, $\lVert a \rVert \leq 1$, $\lVert b \rVert \leq 1$ and $\lVert c_k \rVert \leq 1$, for $k=1,2,\ldots,n$.

Since the respective local states on $\Omega_A$ and $\Omega_B$ are $\omega_A\equiv \phi u=
\bigl(\begin{smallmatrix}1\\ a
\end{smallmatrix} \bigr)$ and $\omega_B\equiv \phi^\text{t} u= \bigl(\begin{smallmatrix}
1\\ b
\end{smallmatrix} \bigr)$, where $\text{t}$ denotes transposition, it follows from the definition of $\Omega$ that $\lVert a \rVert \leq 1$, $\lVert b \rVert \leq 1$.

We show that $\lVert c_k \rVert \leq 1$, for $k=1,2,\ldots,n$. Let $v_k\in\mathbb{R}^n$ be a column vector whose $k$th entry is equal to unity and all other entries are zero. It follows that $Cv_k=c_k$. Let $\lVert c_k\rVert>0$. Consider the following extremal product effects: $E_k^{(0)}\equiv \frac{1}{2}\bigl(\begin{smallmatrix}
1\\ -\hat{c}_k
\end{smallmatrix} \bigr)\otimes \frac{1}{2}\bigl(\begin{smallmatrix}
1\\ v_k
\end{smallmatrix} \bigr)$ and $E_k^{(1)}\equiv \frac{1}{2}\bigl(\begin{smallmatrix}
1\\ \hat{c}_k
\end{smallmatrix} \bigr)\otimes \frac{1}{2}\bigl(\begin{smallmatrix}
1\\ -v_k
\end{smallmatrix} \bigr)$, where $\hat{c}_k\equiv\frac{c_k}{\lVert c_k \rVert}$. Using expression (\ref{eq:11}) and the property $Cv_k=c_k$, we obtain
\begin{equation}
\label{eq:12}
E_k ^{(\lambda)} \cdot\phi=\frac{1}{4}\Bigl(1-\lVert c_k\rVert -(-1)^{\lambda}a\cdot\hat{c_k}+(-1)^\lambda b_k\Bigr),
\end{equation}
where $b_k$ is the $k$th entry of $b$ and $\lambda\in\lbrace 0,1\rbrace$. Since $\phi\in\Omega_A\otimes_{\text{max}}\Omega_B$ and $E_k^{(0)},E_k^{(1)}$ are product effects, we have $E_k^{(\lambda)}\cdot\phi\geq 0$, for $\lambda\in \lbrace 0,1\rbrace$. It follows from (\ref{eq:12}) that $2\bigl(E_k^{(0)}+E_k^{(1)}\bigr)\cdot \phi=1 -\lVert c_k\rVert\geq 0$, which implies $\lVert c_k\rVert \leq 1$.
\end{proof}

\begin{lemma}
\label{lemma3}
Let $\mathcal{E}_{AB}$ be the space for normalised effects corresponding to a bipartite state space $\Omega_{AB}$ that satisfies \hyperref[NS]{no-signalling} and \hyperref[TL]{tomographic locality}, with $\Omega_A\simeq\Omega_B\simeq\Omega$ being a HST system of dimension $n$. Any effect $E\in\mathcal{E}_{AB}$ can be expressed as a matrix
\begin{equation}
\label{eq:13}
 E=\gamma\begin{pmatrix}
1 & \beta^{\text{t}} \\
 \alpha  & \Gamma
 \end{pmatrix},
\end{equation}
where $0\leq\gamma\leq 1$, $\alpha,\beta\in\mathbb{R}^n$, $\Gamma\in\mathbb{R}^n\otimes\mathbb{R}^n$, $\lVert\alpha\rVert\leq 1$, $\lVert\beta\rVert\leq 1$, $\lVert \gamma_k\rVert \leq 1$, with $\gamma_k$ being the $k$th column vector of $\Gamma$, for $k=1,2,\ldots,n$.
\end{lemma}

\begin{proof}
The principles of \hyperref[NS]{no-signalling} and \hyperref[TL]{tomographic locality} imply that any state $\phi\in\Omega_{AB}$ can be expressed as the matrix (\ref{eq:11}). Thus, any effect $E\in\mathcal{E}_{AB}$ can be expressed as the matrix given by Eq. (\ref{eq:012}) of the main text:
\begin{equation}
\label{eq:14}
E=
 \begin{pmatrix}
\gamma & \beta^{\text{t}} \\
 \alpha & \Gamma
 \end{pmatrix},
\end{equation}
where $\gamma\in\mathbb{R}$, $\alpha,\beta\in\mathbb{R}^n$, $\Gamma\in\mathbb{R}^n\otimes\mathbb{R}^n$. We show below the following inequalities:
\begin{eqnarray}
\label{eq:14.1}
0&\leq& \gamma \leq 1,\\
\label{eq:14.2}
\lVert \alpha \rVert &\leq& \min \lbrace \gamma,1-\gamma\rbrace,\\
\label{eq:14.3}
\lVert \beta \rVert &\leq& \min\lbrace \gamma,1-\gamma\rbrace,\\
\label{eq:14.4}
\lVert \gamma_k \rVert&\leq& \min\lbrace\gamma,1-\gamma\rbrace,
\end{eqnarray}
for $k=1,2,\ldots,n$. From (\ref{eq:14}) -- (\ref{eq:14.4}), we see that the case $\gamma=0$ corresponds to the zero effect, which has an expression in agreement with (\ref{eq:13}). For the case $\gamma>0$, it follows from  (\ref{eq:14}) -- (\ref{eq:14.4}) that $E$ can be expressed as the matrix
\begin{equation}
\label{eq:15}
E=\gamma
 \begin{pmatrix}
1 & \beta'^{\text{t}} \\
 \alpha'  & \Gamma'
 \end{pmatrix},
\end{equation}
where $\alpha'\equiv\gamma^{-1}\alpha, \beta'\equiv\gamma^{-1}\beta$ and $\Gamma'\in\mathbb{R}^n\otimes\mathbb{R}^n$ has column vectors $\gamma_k'\equiv\gamma^{-1}\gamma_k$, with $\lVert \alpha'\rVert\leq 1, \lVert \beta'\rVert\leq 1, \lVert\gamma_k'\rVert \leq 1$, for $k=1,2,\ldots,n$. Thus, by dropping the primes in (\ref{eq:15}), we obtain the claimed result (\ref{eq:13}).

We complete the proof by showing (\ref{eq:14.1}) -- (\ref{eq:14.4}). To do so, we use the expression of $E$ given by (\ref{eq:14}). We apply $E$ on particular product states $\phi$ and use the conditions $0\leq E\cdot\phi\leq 1$. This follows because $\Omega_A\otimes_{\text{min}}\Omega_B\subseteq\Omega_{AB}$.

We show (\ref{eq:14.1}). We apply the effect $E$ on the state $\phi_0\equiv \bigl(\begin{smallmatrix}
1\\ \bold{0}
\end{smallmatrix} \bigr)\otimes\bigl(\begin{smallmatrix}
1\\ \bold{0}
\end{smallmatrix} \bigr)$. We obtain $E\cdot\phi_0=\gamma$. The condition $0\leq E\cdot\phi_0\leq 1$ implies (\ref{eq:14.1}).

The inequalities (\ref{eq:14.2}) -- (\ref{eq:14.4}) are trivially satisfied if $\lVert \alpha\rVert = \lVert \beta\rVert = \lVert \gamma_k\rVert = 0$. Thus, below we consider the case $\lVert \alpha\rVert >0$, $\lVert \beta\rVert>0$, $\lVert \gamma_k\rVert>0$, for $k=1,2,\ldots,n$.

We show (\ref{eq:14.2}). Consider the states $\phi_{\pm}\equiv\bigl(\begin{smallmatrix}
1\\ \pm\hat{\alpha}
\end{smallmatrix} \bigr)\otimes\bigl(\begin{smallmatrix}
1\\ \bold{0}
\end{smallmatrix} \bigr)$, where $\hat{\alpha}\equiv\lVert \alpha \rVert ^{-1} \alpha$. We have that $E\cdot\phi_{\pm}=\gamma \pm \lVert \alpha \rVert$. The conditions $0\leq E\cdot\phi_{\pm}\leq 1$ imply (\ref{eq:14.2}).

We show (\ref{eq:14.3}). Consider the states $\phi'_{\pm}\equiv\bigl(\begin{smallmatrix}
1\\ \bold{0}
\end{smallmatrix} \bigr)\otimes \bigl(\begin{smallmatrix}
1\\ \pm\hat{\beta}
\end{smallmatrix} \bigr)$, where $\hat{\beta}\equiv\lVert \beta \rVert ^{-1} \beta$. We have that $E\cdot\phi'_{\pm}=\gamma \pm \lVert \beta \rVert$. The conditions $0\leq E\cdot\phi'_{\pm}\leq 1$ imply (\ref{eq:14.3}).

We show (\ref{eq:14.4}). Let $v_k\in\mathbb{R}^n$ be a column vector with the $k$th entry equal to unity and the other entries equal to zero. We have that $\Gamma v_k =\gamma_k$. Let $\omega_k^{(\mu)}\equiv \bigl(\begin{smallmatrix}
1\\ (-1)^\mu\hat{\gamma}_k
\end{smallmatrix} \bigr)$ and $\tilde{\omega}_k^{(\nu)}\equiv \bigl(\begin{smallmatrix}
1\\ (-1)^\nu v_k
\end{smallmatrix} \bigr)$, where $\hat{\gamma}_k\equiv\lVert \gamma_k\rVert^{-1}\gamma_k$ and $\mu,\nu\in\lbrace 0,1 \rbrace$. Consider the states $\phi_k^{(\mu,\nu)}\equiv \omega_k^{(\mu)}\otimes\tilde{\omega}_k^{(\nu)}$, for $\mu,\nu\in\lbrace 0,1 \rbrace$ and $k=1,2,\ldots,n$. We obtain
\begin{equation}
\label{eq:16}
E\cdot\phi_k^{(\mu,\nu)} =(-1)^{\mu+\nu} \lVert \gamma_k\rVert +(-1)^\mu \alpha\cdot \frac{\gamma_k}{\lVert \gamma_k \rVert} +(-1)^\nu\beta_k +\gamma,
\end{equation}
where $\beta_k$ is the $k$th component of $\beta$. From the conditions $0\leq E\cdot\phi_k^{(\mu,\nu)}\leq 1$, we obtain $1\geq \frac{1}{2}E\cdot\bigl(\phi_k^{(0,0)}+\phi_k^{(1,1)}\bigr)=\lVert \gamma_k \rVert +\gamma$ and $0\leq \frac{1}{2}E\cdot\bigl(\phi_k^{(0,1)}+\phi_k^{(1,0)}\bigr)=-\lVert \gamma_k \rVert +\gamma$, and hence (\ref{eq:14.4}), which completes the proof.
\end{proof}

\begin{proof}[Proof of Proposition \ref{theorem2}]
Since $\Omega_{AB}$ satisfies \hyperref[NS]{the no-signalling principle} and \hyperref[TL]{tomographic locality},  Lemmas \ref{lemma2} and \ref{lemma3} hold. Thus, the states $\phi_x\in\Omega_{AB}$ and the effects $E_y\in\mathcal{E}_{AB}$ can be expressed as in (\ref{eq:11}) and (\ref{eq:13}):
\begin{equation}
\label{eq:a1}
\phi_x= 
 \begin{pmatrix}
1 & b^{\text{t}}_x \\
 a_x  & C_x
 \end{pmatrix},
\qquad E_y=\gamma_y\begin{pmatrix}
1 & \beta^{\text{t}}_y \\
 \alpha_y  & \Gamma_y
 \end{pmatrix}.
\end{equation}
It follows from (\ref{eq:a1}) and Proposition \ref{theorem1A} that
\begin{equation}
\label{eq:a2}
I(X:Y)\leq \log_2 (1+\chi),
\end{equation}
where $\chi\equiv \max \{\chi_{x,y}\}$, $\chi_{x,y}\equiv\alpha_y\cdot a_x+\beta_y\cdot b_x+\Gamma_y\cdot C_x$ and the maximum is taken over all states $\phi_x\in\Omega_{AB}$ and measurements with effects $E_y\in\mathcal{E}_{AB}$. Let this maximum be achieved by a state $\phi$ and an effect $E$, for which we drop the $x$ and $y$ labels. From (\ref{eq:a1}), we have 
\begin{equation}
\label{eq:a3}
E\cdot\phi = \gamma (1+\chi),
\end{equation}
where $\chi\equiv\alpha\cdot a+\beta\cdot b+\Gamma\cdot C$.

The set $\mathcal{T}$ of allowed local transformations on HST systems $\Omega_A\simeq\Omega_B\simeq\Omega$ of dimension $n\geq 2$ that satisfy local continuous reversibility has elements of the form given by Eq. (\ref{eq:a0}), where the group $\hat{\mathcal{T}}$ must be transitive on the unit sphere in $\mathbb{R}^n$ \cite{MMAP13}. There are various groups with this property. In general, $\hat{\mathcal{T}}$ is a subgroup of $\text{SO}(n)$. For odd $n\neq 7$, the only possibility is $\hat{\mathcal{T}}=\text{SO}(n)$ \cite{MMPA14}.

 From the \hyperref[consistency]{consistency} condition, we have $T\phi_x T'^{\text{t}}\in\Omega_{AB}$ for all $\phi_x\in\Omega_{AB}$ and all $T,T'\in\mathcal{T}$, hence, $E_y \cdot( T\phi_x T'^{\text{t}} )\geq 0$ for all $E_y\in \mathcal{E}_{AB}$. Thus, $E\cdot (T\phi T'^{\text{t}} )\geq 0$ for all $T,T'\in\mathcal{T}$. It follows from (\ref{eq:a3}) that
\begin{equation}
\label{eq:a4}
1+\alpha\cdot a'+\beta\cdot b'+\Gamma \cdot C' \geq 0,
\end{equation}
where $a'\equiv \hat{T}a$, $b'\equiv \hat{T'}b$ and $C'\equiv \hat{T} C\hat{T'}^{\text{t}}$, for all $\hat{T},\hat{T'}\in\hat{\mathcal{T}}$.

We show i). As said above, for odd $n\neq 7$, local continuous reversibility implies that $\hat{\mathcal{T}}=\text{SO}(n)$. Consider Eq. (\ref{eq:a4}) for the case $a=\bold{0}$. Let $\hat{T}=I$ and $\hat{T'}=Q_k$ being a diagonal matrix with all entries equal to $-1$ except for the $k$th entry, which equals unity. We have that $\hat{T},\hat{T'}\in\text{SO}(n)$. It follows that
\begin{equation}
\label{eq:a5}
1-\beta\cdot b- \Gamma\cdot C +2(\beta_k b_k+\gamma_k\cdot c_k)\geq 0,
\end{equation}
where $b_k$ and $\beta_k$ are the $k$th entries of $b$ and $\beta$, and $c_k$ and $\gamma_k$ are the $k$th column vectors of $C$ and $\Gamma$, respectively. Thus, in the case $a=\bold{0}$, we have from (\ref{eq:a5}) that
\begin{equation}
\label{eq:a6}
\chi =\beta\cdot b+\Gamma\cdot C \leq 1 + 2(\beta_k b_k+\gamma_k\cdot c_k),
\end{equation}
for $k=1,2\ldots,n$. Since $\beta\cdot b=\sum_{k=1}^n\beta_k b_k$ and $\Gamma\cdot C = \sum_{k=1}^n \gamma_k\cdot c_k$, it follows from (\ref{eq:a6}) that 
\begin{equation}
\label{eq:a7}
\chi \leq \frac{1}{1-\frac{2}{n}}\leq 3,
\end{equation}
where the second inequality is achieved only for $n=3$. Thus, from (\ref{eq:a2}) and (\ref{eq:a7}), we have $I(X:Y)\leq 2$, with equality achieved only for $n=3$. The case $b=\bold{0}$ is proved similarly by considering $\hat{T'}=I$ and $\hat{T}=Q_k$.

We show ii). Since we assume that $\hat{\mathcal{T}}=\text{SO}(n)$ for arbitrary odd $n$, the proof follows straightforwardly as for i).

We show iii). We assume $n$ even and $\hat{\mathcal{T}}=\text{SO}(n)$. Consider equation (\ref{eq:a4}) for the case $\hat{T}=I$ and $\hat{T'}=-I$. We have $\hat{T},\hat{T'}\in\text{SO}(n)$. It follows that
\begin{equation}
\label{eq:a8}
1+\alpha\cdot a-\beta\cdot b-\Gamma\cdot C \geq 0.
\end{equation}
Thus, we have
\begin{equation}
\label{eq:a9}
\chi \leq 1 + 2\alpha\cdot a.
\end{equation}
If instead we have $\hat{T}=-I$ and $\hat{T'}=I$, we obtain the similar expression
\begin{equation}
\label{eq:a10}
\chi \leq 1 + 2\beta\cdot b.
\end{equation}
From Lemmas \ref{lemma2} and \ref{lemma3}, we have $\lVert a \rVert \leq 1$, $\lVert b \rVert \leq 1$, $\lVert \alpha \rVert \leq 1$ and $\lVert \beta \rVert \leq 1$. It follows from (\ref{eq:a2}), (\ref{eq:a9}) and (\ref{eq:a10}) that $I(X:Y)\leq 2$ and $I(X:Y)\leq 1$ if $a=\bold{0}$ or $b=\bold{0}$.

We show iv). We assume $\hat{\mathcal{T}}=\text{O}(n)$ and arbitrary $n$.  Consider equation (\ref{eq:a4}) for the case $\hat{T}=I$ and $\hat{T'}=-I$. We have $\hat{T},\hat{T'}\in\text{O}(n)$. The proof follows similarly as for iii).

We show v). We assume $\hat{\mathcal{T}}=\text{SO}(n)$ and even $n$, or $\hat{\mathcal{T}}=\text{O}(n)$ and arbitrary $n$. In any case we have $-I\in\hat{\mathcal{T}}$. Considering equation (\ref{eq:a4}) with $\hat{T}=I$ and $\hat{T'}=-I$, we obtain, as in the proof of iii), the expressions (\ref{eq:a9}) and (\ref{eq:a10}), from which it follows that $I(X:Y)\leq 2$. However, the equality $I(X:Y)=2$ cannot be achieved by states that are obtained by local transformations, $\phi_x=T_x\phi_0$, as we show. On the one hand, from (\ref{eq:a2}) and (\ref{eq:a9}), a necessary condition for the equality $I(X:Y)=2$ is that $\lVert a \rVert =1$, that is, that the local state $\omega_a\in\Omega_A$ is pure. On the other hand, if $\omega_a$ is pure, it is easy to see that $\phi$ must be product: $\phi=\omega_a\otimes \omega_b$, for some $\omega_b\in\Omega_B$ \cite{BW13}. Thus, if $\lVert a \rVert =1$, the bipartite states are product: $\phi_x=\omega_{a_x}\otimes \omega_b$. It follows from Proposition \ref{theorem0} that in this case $I(X:Y)\leq 1 <2$.
\end{proof}

\section{Locally continuous hyperdense coding violating tomographic locality}
\label{app:lcrtlhdc}

We show that relaxing \hyperref[TL]{tomographic locality} allows hyperdense coding while still satisfying \hyperref[LCT]{local continuous reversibility}. 

Consider local state spaces $\Omega_A\simeq\Omega_B\simeq\Omega_m^{(n)}$, where $\Omega_m^{(n)}\equiv\Bigl\{\omega_r\equiv\Bigl(\begin{smallmatrix}
1\\ \bold{0}\\ r
\end{smallmatrix} \Bigr)\Big\vert \bold{0}\in\mathbb{R}^{n}, r\in\mathbb{R}^m, \lVert r \rVert\leq 1\Bigr\}$, for $m\in\mathbb{N}$ and $n\in\mathbb{Z}_+$. The pure states satisfy $\lVert r \rVert =1$. The unit effect is $u=\bigl(\begin{smallmatrix}
1\\ \bold{0}
\end{smallmatrix} \bigr)$ with $\bold{0}$ being the null vector in $\mathbb{R}^{n+m}$. The state space $\Omega_m^{(n)}$ corresponds to an $m-$hypersphere system embedded in a bigger vector space. The space of local effects $\mathcal{E}$ is the convex hull of the zero effect, the unit effect $u$ and the extremal effects $e_{r}\equiv \frac{1}{2}\Bigl(\begin{smallmatrix}
1\\ \bold{0}\\ r
\end{smallmatrix} \Bigr)$, where $\bold{0}$ is the null vector in $\mathbb{R}^{n}$ and $r\in\mathbb{R}^m$ with $\lVert r \rVert =1$. In what follows we consider $n=2^{N}-1$, with $N\in\mathbb{N}/\{1\}$. The case $N=1$ corresponds to a classical bit.

We define the set of allowed local transformations as $\mathcal{T}_A\simeq\mathcal{T}_B\simeq\mathcal{T}$. The elements of $\mathcal{T}$ are defined by
\begin{equation}
\label{eq:b1}
T_\mu^{(R)}= 
 \begin{pmatrix}  T_\mu &  0\\
  0  & R
 \end{pmatrix},
\end{equation}
where $R\in\text{SO}(m)$ and $T_\mu\in\mathcal{T}_N$, for $\mu\in\{0,1\}^N$, as defined by Eq. (\ref{eq:43.51}) of the main text. It is straightforward to see that the local state space $\Omega_m^{(n)}$ remains invariant under the set of allowed local transformations $\mathcal{T}$. 

Consider the following entangled states
\begin{equation}
\label{eq:b2}
\Phi_\mu\equiv 
 \begin{pmatrix} \phi_\mu &  0\\
  0  & 0
 \end{pmatrix},
\end{equation}
where $\phi_\mu\in\mathcal{S}_N$ is defined by Eq. (\ref{eq:43.310}) of the main text, for $\mu\in\{0,1\}^N$. That is, $\Phi_\mu$ is a state given by a diagonal matrix whose first $n+1=2^N$ entries correspond to the diagonal matrix $\phi_\mu$ and the last $m$ entries are zero. 

We define $\Omega_{AB}\equiv\text{convex hull}\{\Omega_A\otimes_{\text{min}}\Omega_B, \Lambda\}$, 
where $\Lambda\equiv \{T\Phi_{\bold{0}}T'^{\text{t}}\vert T,T'\in\mathcal{T}\}$ and where $\Phi_{\bold{0}}$ is the state $\Phi_{\mu}$ with $\mu=(0,0,\ldots,0)$ being the string of $N$ zero entries. 
From Eqs. (\ref{eq:43.51}) and (\ref{eq:43.52}) of the main text, and Eqs. (\ref{eq:b1}) and (\ref{eq:b2}), it is straightforward to see that $\Lambda=\{\Phi_\mu\}_{\mu\in\{0,1\}^N}$.

We show that $\Omega_{AB}\subseteq\Omega_A\otimes_{\text{max}}\Omega_B$. From the definition of $\Omega_{AB}$, we only need to show that $\Lambda\subset\Omega_A\otimes_{\text{max}}\Omega_B$. Consider arbitrary effects $e=\Bigl(\begin{smallmatrix}
\chi\\ \bold{0}\\ \alpha
\end{smallmatrix} \Bigr)\in\mathcal{E}_A$ and $f=\Bigl(\begin{smallmatrix}
\xi\\ \bold{0}\\ \beta
\end{smallmatrix} \Bigr)\in\mathcal{E}_B$. We have that $\chi,\xi\in[0,1]$ and $\alpha,\beta\in\mathbb{R}^m$. Let $\Phi_\mu\in\Lambda$. We have that 
\begin{equation}
\label{eq:b3}
(e\otimes f)\cdot \Phi_\mu=\chi\xi.
\end{equation}
Since $\chi,\xi\in[0,1]$, we have $(e\otimes f)\cdot \Phi_\mu\in[0,1]$ for all $\mu\in\{0,1\}^N$. We also have that $(u_A\otimes u_B) \cdot \Phi_\mu =1$. Thus, $\Lambda\subset\Omega_A\otimes_{\text{max}}\Omega_B$ and $\Omega_{AB}\subseteq\Omega_A\otimes_{\text{max}}\Omega_B$.

We show that the set of allowed local transformations $\mathcal{T}$ satisfies the \hyperref[consistency]{consistency} condition. As noticed above, the local state space $\Omega_m^{(n)}$ remains invariant under $\mathcal{T}$. By definition, $\Lambda$ remains invariant under local transformations from $\mathcal{T}$ too. Thus, from the definition of $\Omega_{AB}$, we have $T\omega\in\Omega_A$ and $T\Phi\in\Omega_{AB}$ for all $\omega\in\Omega_A$, $\Phi\in\Omega_{AB}$ and $T\in\mathcal{T}$, and similarly for system $B$, as required.

It is easy to see that local continuous reversibility is satisfied. Consider any pair of pure states $\omega_{r_0},\omega_{r_1}\in\Omega_m^{(2^N-1)}$. There exist continuous reversible transformations $T_\mu^{(R)},T_\mu^{(R^{\text{t}})}\in\mathcal{T}$ such that $Rr_0=r_1$ (and $R^{\text{t}}r_1=r_0$) for any $\mu\in\{0,1\}^N$, hence, $T_\mu^{(R)}\omega_{r_0}=\omega_{r_1}$ and $T_\mu^{(R^{\text{t}})}\omega_{r_1}=\omega_{r_0}$. 

Similarly, any pair of pure entangled states $\Phi_\mu,\Phi_{\mu\oplus\mu'}\in\Lambda$ are connected by reversible transformations: $\Phi_{\mu\oplus\mu'}=T_{\mu'}^{(R)}\Phi_\mu$, $\Phi_\mu=T_{-\mu'}^{(R)}\Phi_{\mu\oplus\mu'}$, as follows from Eqs. (\ref{eq:b1}), (\ref{eq:b2}) and Eq. (\ref{eq:43.52}) of the main text. However, these transformations do not act continuously on the states $\Phi_\mu,\Phi_{\mu\oplus\mu'}$, which is easily seen from the fact that the set $\mathcal{T}_N$ in (\ref{eq:b1}) is discrete.

On the other hand, tomographic locality is violated because there exist states in $\Omega_{AB}$, the entangled states in $\Lambda$, that cannot be determined from the outcome probabilities of local measurements performed on $\Omega_A$ and $\Omega_B$. This is easily seen from (\ref{eq:b3}) because the outcome probabilities of local measurements on states $\Phi_\mu\in\Lambda$ are independent of the state $\Phi_\mu$. Thus, the states in $\Lambda$ cannot be determined from local measurements.

We introduce dense coding protocols in the state space $\Omega_{AB}$ defined above that are \hyperref[HDC]{hyperdense}. Alice and Bob initially share the state $\Phi_{\bold{0}}$ given by (\ref{eq:b2}). With probability $p_x=2^{-N}$, Alice implements the local transformation $T_x^{(R)}$ defined by (\ref{eq:b1}) for some $R\in\text{SO}(m)$, and the state transforms into $T_x^{(R)}\Phi_{\bold{0}}=\Phi_x$, for $x\in\{0,1\}^N$. Alice sends Bob her system. Bob applies the joint measurement defined by the effects
\begin{equation}
\label{eq:b4}
F_y\equiv2^{-N}\Phi_y=
 \begin{pmatrix}
E_y &  \bold{0}\\
 \bold{0}  & \bold{0}
 \end{pmatrix},
\end{equation}
for $y\in\{0,1\}^N$, where in the second equality we used the definition (\ref{eq:b2}), and expression (\ref{eq:43.FN}) of the main text. 

We show that $\{F_y\}_{y\in\{0,1\}^N}$ defines a measurement. First, we show that $\{F_y\}_{y\in\{0,1\}^N}\subset\mathcal{E}_{AB}$. Since $\Omega_{AB}\equiv \text{convex hull} \{\Omega_A\otimes_{\text{min}}\Omega_B, \Lambda\}$ with $\Lambda=\{\Phi_x\}_{x\in\{0,1\}^N}$, we only need to show that $0\leq F_y\cdot (\omega_a\otimes\omega_b)\leq 1$ and  $0\leq F_y\cdot \Phi_x\leq 1$, for all $x,y\in\{0,1\}^N$, $\omega_a\in\Omega_A$ and $\omega_b\in\Omega_B$. It is easy to see that $F_y \cdot (\omega_a\otimes\omega_b)=2^{-N}$ for all $y\in\{0,1\}^N$, $\omega_a\in\Omega_A$ and $\omega_b\in\Omega_B$. Moreover, from (\ref{eq:b2}), (\ref{eq:b4}) and Eq. (\ref{eq:43.33}) of the main text, we have
\begin{equation}
\label{eq:b5}
F_y\cdot\Phi_x= E_y\cdot\phi_x=\delta_{y,x},
\end{equation}
for all $x,y\in\{0,1\}^N$. Thus, $\{F_y\}_{y\in\{0,1\}^N}\subset\mathcal{E}_{AB}$. Second, from (\ref{eq:b4}) and Eq. (\ref{eq:lasteqever}) of the main text, we have $\sum_{y\in\{0,1\}^N}F_y=\text{diag}(1,\bold{0})$, with $\bold{0}$ being the null vector in $\mathbb{R}^{n+m}$ and $n=2^N-1$, which is the unit effect $u_{AB}$ on $\Omega_{AB}$.

This protocol achieves a mutual information between Alice's and Bob's random variables $X$ and $Y$ of
\begin{equation}
\label{eq:b6}
I(X:Y) = N,
\end{equation}
as follows from (\ref{eq:b5}). Since the classical capacity of the systems in $\Omega_A$ or $\Omega_B$ is one bit for any value of $m$, this dense coding protocol is \hyperref[HDC]{hyperdense} for $N>2$.

\section{Hypersphere theories violating local continuous reversibility that do not have hyperdense coding or superdense coding}
\label{lastapp}
We give an example of a class of nontrivial theories with local systems being HST systems that satisfy tomographic locality, violate local reversibility and \hyperref[LCT]{local continuous reversibility}. Thus, in general, the given theories violate continuous reversibility. We show that some of these theories do not have \hyperref[HDC]{hyperdense coding}, while others do not even have \hyperref[SDC]{superdense coding}.

Consider a theory of two HST systems of dimension $n=2^{N}-1$ and $N\geq 2$, with a joint state space $\Omega_{AB}=\text{convex hull}\{\Omega_A\otimes_{\text{min}}\Omega_B,\mathcal{S}_N^{(\lambda)}\}$. The set of entangled states is defined by $\mathcal{S}_N^{(\lambda)}\equiv\{\phi_\mu^{(\lambda)}\}_{\mu\in\{0,1\}^N}$, where the states $\phi_\mu^{(\lambda)}$ are given by Eq. (\ref{eq:47}). That is, we have 
 \begin{equation}
\label{eq:G1}
\phi_\mu^{(\lambda)}=
 \begin{pmatrix}
1 & \bold{0} \\
 \bold{0}  & \lambda\hat{T}_\mu
 \end{pmatrix},
\end{equation}
where $\hat{T}_\mu$ is defined as in Eq. (\ref{eq:43.53}) of the main text, for $\mu\in\{0,1\}^N$ and some $\lambda\in[-1,1]$ that we specify below.

We note that these theories satisfy tomographic locality and violate local reversibility and local continuous reversibility, as the theories defined in Section \ref{secht}. Since we have a discrete set of entangled states, and the sets of local pure states are hyperspheres, the consistency condition implies that not all transformations that connect the whole set of local pure states are allowed local transformations. Thus, the set of local pure states cannot be connected, which violates local reversibility and local continuous reversibility.

Now we show that there are numbers $\bar{\lambda}_0,\bar{\lambda}_1$ satisfying $0<\bar{\lambda}_0<\bar{\lambda}_1<1$ such that there is not hyperdense coding if $\lvert \lambda\rvert\leq\bar{\lambda}_1$, and there is not superdense coding if $\lvert \lambda\rvert\leq\bar{\lambda}_0$. First, we use Lemma \ref{lemma3} stated in \hyperref[app:theorem]{Appendix \ref{app:theorem}} to write an arbitrary effect $E$ in the form
\begin{equation}
\label{eq:G2}
E\equiv 
 \gamma\begin{pmatrix}
1 & \beta^{\text{t}} \\
 \alpha  & \Gamma
 \end{pmatrix},
\end{equation}
where $0\leq\gamma\leq 1$, $\alpha,\beta\in\mathbb{R}^n$, $\Gamma\in\mathbb{R}^n\otimes\mathbb{R}^n$, $\lVert\alpha\rVert\leq 1$, $\lVert\beta\rVert\leq 1$, $\lVert \gamma_k\rVert \leq 1$, with $\gamma_k$ being the $k$th column vector of $\Gamma$, for $k=1,2,\ldots,n$.

Second, we note from the convexity of the mutual information \cite{CTbook} and the definition of $\Omega_{AB}$ that a dense coding protocol that achieves the \hyperref[TDCC]{dense coding capacity} of the theory must use a pure entangled state $\phi_{\mu}^{(\lambda)}\in\mathcal{S}_N^{(\lambda)}$ as an initial state.  This is because if the initial state is a convex combination of several states, being these entangled or product, the \hyperref[DCC]{dense coding capacity of the state} can only decrease. 

Third, using (\ref{eq:G1}), (\ref{eq:G2}) and Proposition \ref{theorem1A}, we bound the dense coding capacity of the theory by 
\begin{equation}
\label{eq:G3}
\chi_{\text{DC}}(\Omega_{AB})\leq \log_2\bigl( 1+\max \{\lambda \hat{T}_\mu\cdot \Gamma\}\bigr).
\end{equation}

Fourth, since the matrix $\hat{T}_{\mu}$ is diagonal with entries either $1$ or $-1$, and the diagonal entries of $\Gamma$ are bounded by $-1$ and $1$, as follows from the bound $\lVert \gamma_k\rVert \leq 1$ in (\ref{eq:G2}), we obtain from (\ref{eq:G3}) that
\begin{equation}
\label{eq:G4}
\chi_{\text{DC}}(\Omega_{AB})\leq \log_2\bigl( 1+\lvert \lambda\rvert(2^N-1)\bigr).\\
\end{equation}

Finally, if we set $\bar{\lambda}_j\equiv \frac{1+2j}{2^N-1}$ for $j=0,1$, we obtain $\chi_{\text{DC}}(\Omega_{AB})\leq 2$ if $\lvert \lambda\rvert\leq\bar{\lambda}_1$, and $\chi_{\text{DC}}(\Omega_{AB})\leq 1$ if $\lvert \lambda\rvert\leq\bar{\lambda}_0$. Thus, it follows from the definitions of hyperdense coding and superdense coding, and the fact that the classical capacities of HSTs is $1$ bit that for the class of theories given above,
there cannot be hyperdense coding if $\lvert \lambda\rvert\leq\bar{\lambda}_1$, and there cannot be superdense coding if $\lvert \lambda\rvert\leq\bar{\lambda}_0$, as claimed.

%\section{Discussion}


%\bibliography{Hyperspheres}
\end{document}